\documentclass[reprint,preprintnumbers,nofootinbib,superscriptaddress,amsmath,amssymb,aps,prd]{revtex4-2}
\pdfoutput=1
\usepackage{graphicx,hyperref,braket,orcidlink, algorithm2e, amssymb, amsmath, dsfont, subfigure, lineno, bm, slashed, multirow, amsthm, booktabs}
\usepackage[nolist,nohyperlinks]{acronym}

\graphicspath{{figs/}}
\hypersetup{
	bookmarks=true,
	unicode=true,
	pdftoolbar=true,
	pdfmenubar=true,
	pdffitwindow=false,
	pdfstartview={FitH},
	pdftitle={pdf},
	pdfauthor={J.Y.~Araz}, 
	pdfsubject={pdf},
	pdfcreator={J.Y.~Araz},
	pdfproducer={}, 
	pdfkeywords={},
	pdfnewwindow=true,
	colorlinks=true,
	linkcolor=blue,
	citecolor=olive,
	filecolor=magenta,
	urlcolor=cyan
}
\definecolor{org}{HTML}{FF4F00}

\newtheorem{theorem}{Theorem}

\newtheorem{definition}{Definition}
\newtheorem{proposition}{Proposition}

\newcommand{\Prob}{\mathbb{P}}
\newcommand{\E}{\mathbb{E}}

\newcommand{\FDP}{\mathrm{FDP}}
\newcommand{\FDR}{\mathrm{FDR}}
\newcommand{\pw}{\hat p^{\,w}}

\begin{document}

\title{Conformal calibration and look-elsewhere effect in anomaly detection\\for new-physics searches}

\author{Jack Y. Araz\orcidlink{0000-0001-8721-8042}}
\email{j.araz@ucl.ac.uk}
\affiliation{Department of Physics and Astronomy, University College London, London, WC1E 6B, UK}
\affiliation{Department of Engineering, City St. George’s, University of London, London, EC1V 0HB, UK}

\author{Michael Spannowsky\orcidlink{0000-0002-8362-0576}}
\email{michael.spannowsky@kit.edu}
\affiliation{Institute for Theoretical Physics, Campus S\"ud, Karlsruhe Institute of Technology (KIT), D-76128 Karlsruhe, Germany}
\affiliation{Institute for Quantum Materials and Technologies, Karlsruhe Institute of Technology, Karlsruhe 76131, Germany}


\begin{abstract}
Machine-learned anomaly detection is reshaping searches for new physics, but it has outrun the statistics used to interpret it. A raw anomaly score has no calibrated meaning, a model that scans many regions inflates the look-elsewhere effect, and the asymptotic significances the field relies on are blind to the background mismodelling that anomaly detectors are especially prone to. We propose a calibration layer, built on conformal prediction, that turns \emph{any} anomaly score into a defensible significance with distribution-free, finite-sample guarantees. Conformal prediction converts scores into valid local $p$-values, weighted and Mondrian variants repair the sideband-to-signal-region exchangeability failures that resonant searches suffer, and a Gross--Vitells step carries the result through to a look-elsewhere-aware global significance. The layer does two things at once. It \emph{exposes} miscalibration that the standard pipeline cannot see, and it \emph{corrects} it, without retraining the detector. On public LHC Olympics data, a classifier develops a substructure--mass correlation that makes sideband-calibrated background $p$-values anti-conservative. 
Taken at face value, this manufactures a $\sim$$46\sigma$ excess from background sculpting alone, which the label-free weighted correction removes, restoring an honest null. 
When run as a blind wide-mass bump hunt, the standard asymptotic and unweighted procedures fabricate $\gtrsim10\sigma$ excesses and $\approx5\sigma$ excesses even in signal-free windows, while the conformal layer raises no false alarms and its global false-positive rate is verified on background-only pseudoexperiments. The result is an auditable, detector-agnostic path from an uncalibrated score to a trials-factor-aware significance, ready to be folded into experimental anomaly searches.
\end{abstract}

\maketitle

\acrodef{QCD}{Quantum chromodynamics}
\newcommand{\QCD}{\ac{QCD}\xspace}

\tableofcontents

\section*{Introduction}

Machine-learned anomaly detection (AD) has become a leading strategy in searches for physics beyond the Standard Model~\cite{Collins2018, Nachman2020, DAgnolo:2018cun, Farina:2018fyg, Kasieczka2021, ATLAS2020cwola, Feickert:2021ajf}. A model trained to flag events or regions that differ from the background emits a scalar \emph{anomaly score}, with larger values indicating greater ``unusualness.'' Two obstacles separate such a score from a defensible discovery statement. First, a raw score is uncalibrated. A value of $17.5$ says nothing by itself about how often the background alone yields something at least as extreme. Second, a flexible model scans many regions, observables, final states, and latent directions; the more places it may look, the more background fluctuations it selects that mimic the signal. This is the look-elsewhere effect~\cite{Gross2010, Cowan2011} in an acute, high-multiplicity form, recently shown to badly miscalibrate AD $p$-values when training and evaluation share data~\cite{Hein2025}.

The experimental response to multiplicity is mature. Local significances are computed from asymptotic profile-likelihood formulae~\cite{Cowan2011}, and a trials factor for a continuous scan is obtained from the Gross--Vitells up-crossing theory~\cite{Gross2010, Vitells2011}, with modern Gaussian-process and Gaussian-random-field accelerations~\cite{Ananiev2023, GRF2023}. Every quoted significance in the experimental AD literature~\cite{ATLAS2020cwola, ATLAS2023ad1, ATLAS2023ad2, ATLAS2025dijet, ATLAS2025svj, ATLAS2025multilep, CMS2024ad} ultimately rests on this asymptotic-plus-trials-factor machinery.

Conformal prediction is built for exactly this gap, producing finite-sample, distribution-free $p$-values that assume no background model, and a mature toolbox has grown around it for outlier and novelty detection. Our starting point is the conformal outlier-testing framework of Ref.~\cite{Bates2023}, which constructs marginally valid conformal $p$-values, establishes their positive dependence, and couples them to Benjamini--Hochberg false-discovery-rate control. That framework rests on the inductive-conformal $p$-value~\cite{ShaferVovk2008, Vovk2005, Vovk2013mondrian, Lei2018} and sits within a fast-moving literature, from AdaDetect~\cite{Marandon2024} and conformal selection~\cite{Jin2023} to integrative~\cite{Liang2024}, cross-conformal~\cite{Hennhofer2024}, and transductive~\cite{Gazin2023} variants. In a previous study, we introduced conformal prediction as a calibration standard in Machine Learning applications for HEP~\cite{Araz:2025vuw}. What has been missing is the bridge from this statistical machinery to new-physics searches at colliders. We supply it by carrying these tools through the trials-factor layer and adapting them to the failure modes a resonant anomaly search actually exhibits, so that an uncalibrated score becomes a global significance a search can quote.

We particularly focus on the canonical \emph{bump hunt}: a search for a new resonance that would appear as a localised excess over a smoothly falling background in some resonant variable, typically an invariant mass. Because the background shape is difficult to predict from first principles, it is estimated from the data. One designates a \emph{signal region} (SR), a window in the resonant variable where the resonance is hypothesised to sit, and \emph{sidebands} (SB) (or control regions), the adjacent signal-depleted intervals that fix the background expectation inside the SR. This control-versus-signal-region logic is the one-dimensional analogue of the ABCD method~\cite{Kasieczka:2020pil}, in which two approximately independent discriminating variables partition the plane into one signal-rich region and three control regions, and the background in the former is predicted from the latter without recourse to simulation. In both cases, the analysis is only as trustworthy as the assumption that the control regions track the background of the signal region.

Machine-learned anomaly detection adds a layer to this scaffold. Methods such as CWoLa hunting~\cite{Collins2018, Collins2019cwola}, ANODE~\cite{Nachman2020}, CATHODE~\cite{Hallin2021cathode}, laCATHODE~\cite{Hallin:2022eoq} or CURTAINS~\cite{Raine:2022hht} train a classifier or density model to separate SR-like events from the sideband-derived background, producing a per-event \emph{anomaly score} that concentrates any signal into its high-score tail. A bump hunt is then run on the score-enhanced sample, and its significance is read from the asymptotic-plus-trials-factor machinery. Such constructions have already been deployed on collider data~\cite{ATLAS2020cwola, ATLAS2025dijet, CMS2024ad}. 
Throughout the methodology, we therefore refer to a sideband calibration sample and a signal-region test sample, where the sideband provides the reference scores and the SR provides the events to be assessed for anomaly.

We close with a worked example on public data (Sec.~\ref{sec:lhco}), the LHC Olympics R\&D dataset~\cite{Kasieczka2021}, with the anomaly score produced by a classifier, similar to ref.~\cite{Collins2018, Collins2019cwola}.\footnote{The procedure applies to any similar analysis and is not limited to a particular machine-learning ansatz. Since the conformal layer acts on the score, a given detector propagates through the identical machinery.} There, the conformal diagnostic exposes a real exchangeability failure that the standard pipeline does not flag; the learned features are typically correlated with the observable that separates the signal and sideband regions, so sideband-calibrated background $p$-values become anti-conservative. We outline necessary features to diagnose and correct such potential failure in a model-agnostic way with minimal assumptions. We explain each construction in operational terms throughout and stress at the outset that this is a demonstration of methodology.

\section{Methodology}
\label{sec:methodology}

Let $s$ denote an anomaly score, oriented so that larger is more anomalous. Under the background-only hypothesis, the score is a random variable whose distribution is some pushforward of the input distribution through the trained model. For a ``black-box'' neural model, this distribution is generally unknown and has no convenient closed form. Consequently, a threshold rule such as ``flag~$s>15$'' has an unknown false-positive rate, and a list of scores cannot be ranked by statistical surprise. In Fig.~\ref{fig:calibration}, we prepared a mock example where anomaly scores for background events are generated from a gamma distribution, and we mimic the signal as a shift in this distribution. We will utilise the same mock dataset between Figs~\ref{fig:conf-val} -- \ref{fig:asym} for illustration purposes. The upper panel illustrates the difficulty by presenting four lines as potential signal events. When plotted against a background score distribution, they look ``large'', but their tail probabilities are not legible from the raw values.

\begin{figure*}[t]
    \centering
    \subfigure[Upper: a background score distribution with four illustrative scores marked; the raw values do not reveal their tail probabilities. Bottom: the conformal calibration map of Eq.~\eqref{eq:cp}, mapping each score to a conformal $p$-value. Dashed red-line represents $5\%$ threshold.\label{fig:calibration}]{\includegraphics[width=0.45\linewidth]{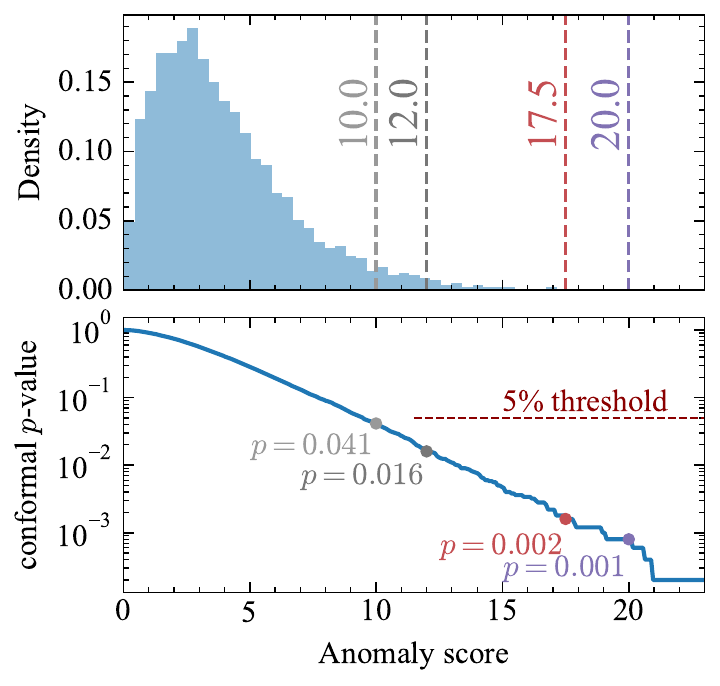}} \qquad\qquad
    \subfigure[Validity check (Prop.~\ref{prop:unif}). Upper: conformal $p$-values on a held-out background sample are approximately uniform. Lower: the empirical exceedance $\Prob(\hat p\le a)$ tracks the diagonal, confirming $\Prob(\hat p\le \alpha)\le \alpha$ up to sampling noise.\label{fig:validity}]{\includegraphics[width=0.45\linewidth]{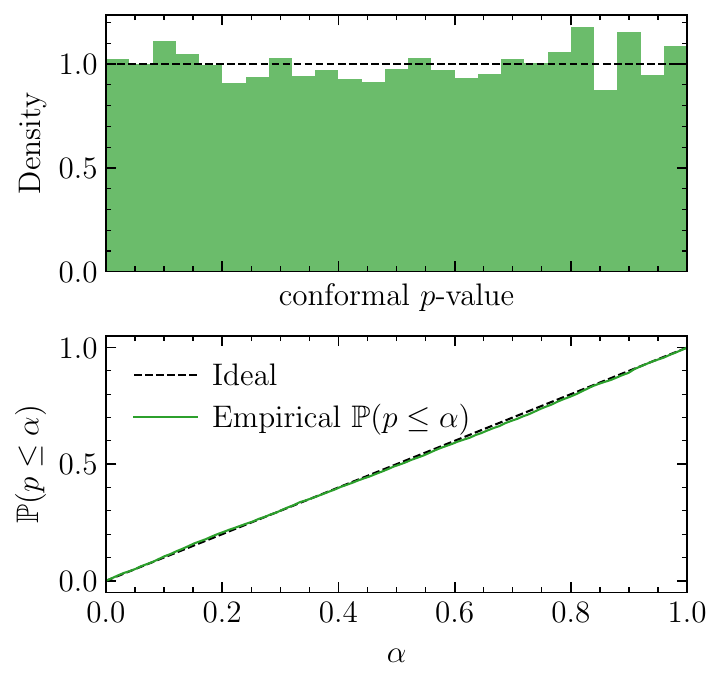}}
    \caption{Conformal calibration and validity checks on mock data.}
    \label{fig:conf-val}
\end{figure*}

\subsection{Split conformal calibration}
\label{sec:cp}

Reserve a \emph{calibration} sample of $n$ background-only scores $s_1,\dots,s_n$, drawn independently of model training. In an experiment, this is a blinded control region or a representative sideband. For a test score $s$, define the one-sided conformal $p$-value~\cite{ShaferVovk2008, Bates2023}
\begin{equation}
  \hat p(s) \;=\; \frac{1 + \bigl|\{\, i : s_i \ge s \,\}\bigr|}{n+1}.
  \label{eq:cp}
\end{equation}
This is the empirical upper-tail mass of the calibration scores at $s$, with an additional $+1$ in the numerator and denominator. The $+1$ smoothing is what makes the guarantee below \emph{exact} in finite samples and prevents the report of a literal zero from a finite calibration set. Eq.~\eqref{eq:cp} is the inductive (split) conformal $p$-value~\cite{Papadopoulos2002, Lei2018}. The score itself serves as the nonconformity measure. The bottom panel of Fig.~\ref{fig:calibration} shows the calibration map of Eq.~\eqref{eq:cp}, mapping each raw score to a conformal $p$-value and providing a quantitative score for each point in the upper panel.

\begin{theorem}[Conformal marginal validity; \cite{Vovk2005, Bates2023}]
\label{thm:cp}
Let $s_1,\dots,s_n,s$ be exchangeable real-valued random variables (in particular, this holds if they are i.i.d.). Let $\hat p(s)$ be defined by Eq.~\eqref{eq:cp}. Then for every $\alpha\in(0,1)$,
\begin{equation}
  \Prob\bigl(\hat p(s) \le \alpha\bigr) \;\le\; \alpha .
  \label{eq:valid}
\end{equation}
\end{theorem}

\begin{proof}
Write $S=(s_1,\dots,s_n,s)$ and set $s_{n+1}:=s$. Consider the rank of $s_{n+1}$ among all $n+1$ values, counting from the top,
\[
  R \;=\; \bigl|\{\, j\in\{1,\dots,n+1\} : s_j \ge s_{n+1} \,\}\bigr|.
\]
By the definition Eq.~\eqref{eq:cp}, and because the test point is included in its own count once (the leading $+1$ in the numerator),
\[
  \hat p(s) \;=\; \frac{R}{n+1}.
\]
Suppose first that ties occur with probability zero. Exchangeability of $S$ implies that the rank of the designated element $s_{n+1}$ is uniformly distributed over $\{1,2,\dots,n+1\}$, that is, $\Prob(R=r)=1/(n+1)$ for each $r$. Hence for any $\alpha\in(0,1)$,
\[
  \Prob\!\left(\hat p(s)\le\alpha\right)
   = \Prob\!\left(R \le (n+1)\alpha\right)
   = \frac{\lfloor (n+1)\alpha \rfloor}{n+1}
   \le \alpha .
\]
When ties occur with positive probability, the standard super-uniform rank lemma applies. Breaking ties uniformly at random (or, equivalently, replacing the rank by its randomised version) restores the exact uniform distribution of the rank, and the deterministic $\ge$ convention of Eq.~\eqref{eq:cp} assigns the designated element the largest rank consistent with its score, which can only increase $\hat p(s)$ relative to the randomised value. The bound in Eq.~\eqref{eq:valid} therefore continues to hold, and the $p$-value is \emph{super}-uniform.
\end{proof}

It is essential to note that the bound holds for \emph{any} finite $n$ and makes \emph{no} assumptions about the shape of the score distribution. This \emph{distribution-free}, \emph{finite-sample} character is what is wanted for an opaque score. However, the guarantee is \emph{marginal}. The probability in Eq.~\eqref{eq:valid} averages over the joint draw of the calibration set and the test point, and does not assert conditional validity given a fixed calibration set, nor local validity within a sub-population. Conditional control is the role of the Mondrian construction of Sec.~\ref{sec:mondrian}, and the dependence this marginal averaging leaves in the count field is treated in Sec.~\ref{sec:nullvar}. The assumption that does the work is the exchangeability of the test point with the calibration sample. We return to it in Sec.~\ref{sec:weighted} and demonstrate it on LHCO data in Sec.~\ref{sec:lhco}.

\begin{proposition}[Background uniformity]
\label{prop:unif}
If $s$ is exchangeable with the calibration sample then $\hat p(s)$ is super-uniform
\[
\forall\,\alpha\in[0,1],\,\Prob(\hat p(s)\le\alpha)\le\alpha\ .
\]
It equals the uniform distribution only up to grid effects (the conformal $p$-value lives on the discrete grid $\{1/(n+1),\dots,1\}$) and, when scores tie with positive probability, the conservative excess from ties; equality is exact only on the conformal grid and in the absence of ties.
\end{proposition}

\begin{proof}
Immediate from Theorem~\ref{thm:cp} applied at every level $\alpha$, together with the observation that $R/(n+1)$ takes values on the grid $\{1/(n+1),\dots,(n+1)/(n+1)\}$ each with probability $1/(n+1)$.
\end{proof}

Prop.~\ref{prop:unif} provides a practical diagnostic tool for the method; \emph{on genuine background, conformal $p$-values \textbf{must} be uniformly distributed}. A departure from uniformity is direct evidence that exchangeability has failed. Fig.~\ref{fig:calibration} (bottom) shows the calibration map (raw score to conformal $p$-value), and Fig.~\ref{fig:validity} confirms uniformity on a held-out background sample. Small fluctuations of the empirical exceedance curve about the diagonal are finite-sample sampling noise of $\mathcal{O}(1/\sqrt{n})$ and not a violation, since Eq.~\eqref{eq:valid} bounds an expectation.

\subsection{The multiplicity problem}
\label{sec:bh}

A search produces not one $p$-value but $m$ of them, one per candidate region. Under the background-only hypothesis, each is approximately uniform, so a fixed per-region cut $\alpha$ flags an expected $\alpha m$ background regions by sheer luck. With $m=200$ and $\alpha=0.05$, this is ten tricksy candidates. The error rate of interest is therefore not the per-region rate but a property of the whole shortlist.

\begin{definition}[FDP and FDR]
For a shortlist of $R$ regions of which $V$ are truly background, the false discovery proportion is $\FDP = V/\max(R,1)$, and the false discovery rate is its expectation $\FDR = \E[\FDP]$.
\end{definition}

\subsection{The Benjamini--Hochberg procedure}

Given $p$-values $p_1,\dots,p_m$ and a target level $q\in(0,1)$, sort them as
$p_{(1)}\le\cdots\le p_{(m)}$, find
\begin{equation}
  k \;=\; \max\Bigl\{\, i : p_{(i)} \le \tfrac{i}{m}\,q \,\Bigr\},
  \label{eq:bh}
\end{equation}
and reject (shortlist) the hypotheses corresponding to $p_{(1)},\dots,p_{(k)}$~\cite{Benjamini1995}. If no $i$ satisfies the inequality, nothing is rejected. The threshold line $\tfrac{i}{m}q$ is strictest for the smallest $p$-value and relaxes linearly; the rule takes the \emph{largest} passing rank, rejecting everything up to it.

\subsection{FDR control}

\begin{theorem}[Benjamini--Hochberg, independent case~\cite{Benjamini1995}]
\label{thm:bh}
Suppose the $p$-values corresponding to the $m_0$ true background regions are independent and each is (super-)uniform under its null, and are independent of the remaining $p$-values. Then the procedure in Eq.~\eqref{eq:bh} satisfies
\begin{equation}
  \FDR \;=\; \E[\FDP] \;\le\; \frac{m_0}{m}\, q \;\le\; q.
  \label{eq:bhbound}
\end{equation}
\end{theorem}

\begin{proof}
Index the regions $1,\dots,m$ and let $\mathcal H_0$ be the set of $m_0$ true background indices. For a region $j$, let $R$ denote the total number of rejections made by the procedure. The procedure rejects region $j$ if and only if $p_j \le \tfrac{R}{m} q$, where $R$ is the realised number of rejections. This is the defining fixed-point property of Eq.~\eqref{eq:bh}. Write the false discovery proportion as
\[
  \FDP \;=\; \sum_{j\in\mathcal H_0} \frac{\mathbf 1\{\text{$j$ rejected}\}}{\max(R,1)} .
\]
Decompose over the possible values $R=r$, $r=1,\dots,m$. On the event $\{R=r\}$, a true null $j$ is rejected exactly when $p_j \le \tfrac{r}{m}q$. Hence
\[
  \E[\FDP]
  = \sum_{j\in\mathcal H_0}\sum_{r=1}^{m}
      \frac{1}{r}\,
      \Prob\!\Bigl(p_j \le \tfrac{r}{m}q,\; R=r\Bigr).
\]
A standard leave-one-out argument uses the independence of $p_j$ from the other $p$-values. Let $R^{(j)}_r$ be the event that the procedure makes $r$ rejections when the $p$-value of region $j$ is set to $0$ (so that $j$ is certainly rejected). Independence and the super-uniformity $\Prob(p_j\le u)\le u$ give, for each true null $j$,
\begin{multline*}
  \sum_{r=1}^m \frac{1}{r}\,
   \Prob\!\Bigl(p_j\le\tfrac{r}{m}q,\;R=r\Bigr) \\
  \;\le\; \sum_{r=1}^m \frac{1}{r}\,\frac{rq}{m}\,
   \Prob\!\bigl(R^{(j)}_r\bigr)
  \;=\; \frac{q}{m}\sum_{r=1}^m \Prob\!\bigl(R^{(j)}_r\bigr).
\end{multline*}
The events $\{R^{(j)}_r\}_{r=1}^m$ are disjoint and exhaustive, so the inner sum is $1$. Therefore,
\begin{equation*}
  \E[\FDP] \;\le\; \sum_{j\in\mathcal H_0}\frac{q}{m}
   \;=\; \frac{m_0}{m}\,q \;\le\; q,
\end{equation*}
which is Eq.~\eqref{eq:bhbound}.
\end{proof}

\begin{theorem}[Positive dependence and the general case~\cite{Benjamini2001, Bates2023}]
\label{thm:by}
If the true-null $p$-values satisfy the positive regression dependence on a subset (PRDS) property, the bound in Eq.~\eqref{eq:bhbound} continues to hold for the unmodified procedure. Under arbitrary dependence, the bound holds if $q$ is replaced by $q/\sum_{i=1}^m i^{-1}$ (the Benjamini--Yekutieli correction~\cite{Benjamini2001}).
\end{theorem}

\begin{proof}[Proof sketch]
The PRDS case is Theorem~1.2 of ref.~\cite{Benjamini2001}; the key step replaces the independence argument above by a monotonicity property of the rejection event under PRDS. The arbitrary-dependence case follows by bounding $\sum_{r} 1/r \le \sum_{i=1}^m 1/i$ and absorbing the factor into the level, yielding the stated correction~\cite{Benjamini2001}.
\end{proof}

Theorem~\ref{thm:bh} is the foundational case. It establishes the false discovery rate bound of Eq.~\eqref{eq:bhbound} and the leave-one-out argument behind it, and it is the result that Theorem~\ref{thm:by} extends. However, its independence hypothesis does not hold for the $p$-values we use. Split-conformal $p$-values computed against a shared calibration set are ranked against the same reference scores and are positively dependent. This positive dependence is not an obstacle to error control but the structure that Theorem~\ref{thm:by} requires. Ref.~\cite{Bates2023} proves that conformal $p$-values built from a common calibration set satisfy the PRDS property, so the unmodified procedure in Eq.~\eqref{eq:bh} still controls the false discovery rate at the bound in Eq.~\eqref{eq:bhbound} for the plain split-conformal construction. We therefore retain Theorem~\ref{thm:bh} as the baseline that fixes the bound and isolates the proof, and take the operative guarantee from Theorem~\ref{thm:by} once the dependence is accounted for. For the weighted, Mondrian, or overlapping-window constructions of Sections~\ref{sec:weighted}--\ref{sec:gv}, where PRDS is not established for the exact estimator, the Benjamini--Yekutieli correction still controls the false discovery rate under arbitrary dependence, at the cost of a $\sum_{i=1}^m i^{-1}$ factor in the level. Control of the false discovery rate is thus available throughout. We treat the Benjamini--Hochberg shortlist as a complementary region-selection tool rather than the headline discovery statistic, since the worked example of Sec.~\ref{sec:lhco} reports a scan significance, and FDR control over the discrete window family is an optional output that bounds the expected fraction of spurious regions on any shortlist drawn from it.

\begin{figure}[th]
  \centering
  \subfigure[A single Benjamini--Hochberg triage at target level $q=0.10$: the sorted conformal $p$-values (grey) cross the Benjamini--Hochberg line $iq/m$ (red) just past the ten true signal regions (green), setting the cut (black dashed); the rule shortlists the signal regions.\label{fig:bh}]{\includegraphics[width=0.95\linewidth]{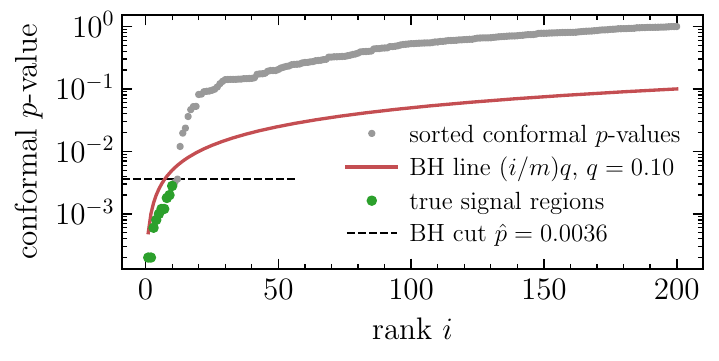}}

  \subfigure[Mean realised false discovery proportion (red) and mean power (blue) over $2000$ experiments per target level, against $q$; the mean FDP lies at or below the line $\mathrm{FDR}=q$ (black dashed) at every level, as Theorems~\ref{thm:bh}--\ref{thm:by} require, while the power saturates once $q$ is large enough to admit the signal regions. The shaded band is $\pm1$ standard deviation of the false discovery proportion across experiments.\label{fig:fdr-power}]{\includegraphics[width=0.95\linewidth]{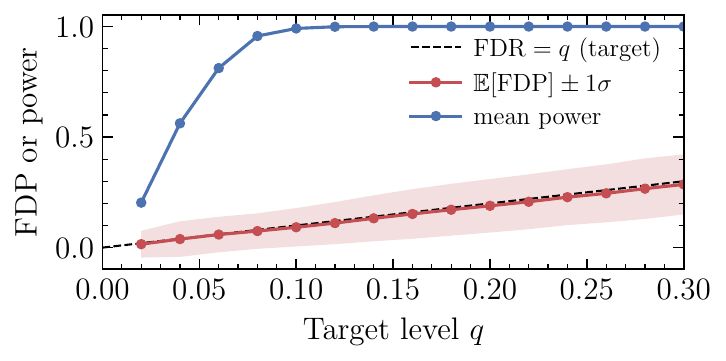}}

  \caption{Benjamini--Hochberg false-discovery-rate control on a mock multi-region search: $m=200$ candidate regions, ten carrying a moderate signal, each scored by a conformal $p$-value against a common background calibration set drawn from the gamma mock model of Fig.~\ref{fig:calibration}.}
  \label{fig:bh-fdr}
\end{figure}

Fig.~\ref{fig:bh-fdr} demonstrates the procedure on a mock multi-region search. We draw $m=200$ candidate regions, ten of which carry a moderate signal (a shift of the gamma score of Fig.~\ref{fig:calibration}) and the rest are background-only, compute one conformal $p$-value per region against a common background calibration set, and apply Eq.~\eqref{eq:bh} at $q=0.1$. A single triage (Fig.~\ref{fig:bh}) shortlists the ten signal regions (green) with two false positives (grey); the points lie underneath the BH line (solid red). Repeating the experiment $2000$ times per target level (Fig.~\ref{fig:fdr-power}) shows the mean false discovery proportion tracking the line $\mathrm{FDR}=q$ from below at every $q$ (red), with the power (blue) rising to one once $q$ admits the signal regions.

This is the setting in which Benjamini--Hochberg is the natural output. Many candidate regions, each a separate discovery hypothesis, from which one wants a shortlist with a controlled false fraction. The worked example of Sec.~\ref{sec:lhco} is not of this kind, which is why we will not apply Benjamini--Hochberg there. It targets a single resonance, so the quantity of interest is a global, look-elsewhere-corrected significance for one excess (Sec.~\ref{sec:gv}), not a false discovery rate over a family of candidates. Moreover, the scan windows there are strongly correlated, and a genuine signal spans several of them, so a shortlist would flag a contiguous block that still corresponds to one resonance. Hence, the BH-FDR procedure provides a generalisation over the methods we will deploy in Sec.~\ref{sec:lhco}.

\subsection{Exchangeability}
\label{sec:weighted}

The conformal guarantee relies on exchangeability between the calibration and test scores. Resonant anomaly searches break it for a reason unrelated to the signal. Such searches single out a \emph{resonant variable} in which a new state would appear as a localised bump, and define a signal region (SR) and sidebands (SB) as intervals in that variable. The anomaly score is built from the \emph{other} event features by a classifier trained to flag SR-like events relative to the SB. Calibration scores are then taken from one region and test scores from another. The difficulty is that, for background, the features feeding the score are generally correlated with the resonant variable, so the SR and SB populate different parts of feature space and the \emph{background} score itself differs between them, $p_{\rm bkg}(s|{\rm SB})\neq p_{\rm bkg}(s|{\rm SR})$, even with no signal present. The weakly-supervised classifier amplifies this. Trained to separate the regions, it latches onto the background features that track the resonant variable. Were the score decorrelated from that variable~\cite{Shimmin:2017mfk, Kasieczka:2020yyl}, the background would coincide across the regions, and exchangeability would hold. However, life is not that simple, and the failure shows up directly in the background, as we confirm on real data in Sec.~\ref{sec:lhco}. Theorem~\ref{thm:cp} then becomes garbage, and Prop.~\ref{prop:unif} merely \emph{detects} the failure. The remedy is to model the failure as a covariate shift and correct it.

Suppose the calibration scores are drawn from $P$, and the test (SR background) scores from $Q$, with a known likelihood ratio $w(x)=\mathrm{d}Q/\mathrm{d}P(x)$. Following ref.~\cite{Tibshirani2019}, we define the weighted conformal
$p$-value
\begin{equation}
  \pw(s) \;=\;
  \frac{w(x^\star) + \sum_{i=1}^n w(x_i)\,\mathbf 1\{s_i \ge s\}}
       {w(x^\star) + \sum_{i=1}^n w(x_i)} ,
  \label{eq:wcp}
\end{equation}
where $x^\star$ is the test covariate (e.g.\ the SR position) and $x_i$ the
calibration covariates. Eq.~\eqref{eq:wcp} replaces the uniform weight
$1/(n+1)$ of Eq.~\eqref{eq:cp} by normalised likelihood-ratio weights.

\begin{theorem}[Validity under covariate shift; \cite{Tibshirani2019}]
\label{thm:wcp}
If the calibration and test scores are \emph{weighted exchangeable} with weight function $w=\mathrm{d}Q/\mathrm{d}P$, then for a test point drawn from $Q$ and any $\alpha\in(0,1)$, $\Prob\bigl(\pw(s)\le\alpha\bigr)\le\alpha$.
\end{theorem}
\begin{proof}
Ref.~\cite{Tibshirani2019} show that under covariate shift, the tuple $(s_1,\dots,s_n,s)$ is exchangeable after reweighting each ordering of the augmented sample by the product of its covariate likelihood ratios. The normalised weights $p_i^w = w(x_i)/\sum_{j} w(x_j)$ (with the test point included) are exactly the probabilities that the test point attains each rank. Summing the weights of the orderings in which the designated point ranks at or above $s$ gives Eq.~\eqref{eq:wcp}, and the same rank-uniformity argument as in Theorem~\ref{thm:cp}, applied to the reweighted measure, yields the bound. For a complete proof, we refer the reader to ref.~\cite{Tibshirani2019} and for the general non-exchangeable bound with estimation error, ref.~\cite{Barber2023shift}.
\end{proof}

Two statements must be kept apart here. Theorem~\ref{thm:wcp} is exact only under weighted exchangeability with the \emph{true} likelihood ratio $w=\mathrm{d}Q/\mathrm{d}P$. It is a statement about the estimator Eq.~\eqref{eq:wcp} fed the correct weights, not about any particular way of obtaining them. ``How on earth would I know the likelihood ratio'', you ask? The procedure we run estimates $w$ from data, which is a separate object with its own coverage properties. Learning the weight from the observed signal region can absorb part of a genuine signal and alter power, an effect we measure directly in Sec.~\ref{sec:lhco_power}. To keep the estimate from contaminating the calibration scores, the weight model should be fit on a sample disjoint from the calibration and test scores, that is, cross-fitted or sample-split. We adopt this split in the worked example. The residual validity cost of an imperfect estimate is then controlled rather than eliminated. Ref.~\cite{Barber2023shift} bounds the coverage gap by the total-variation distance between the reweighted and true test distributions, so a good but imperfect estimate yields a small, bounded validity loss, and we report the empirical coverage of the estimated-weight procedure across calibration sizes and contamination levels in Secs.~\ref{sec:lhco_power} and \ref{sec:lhco_calsize}.

In CATHODE or similar density-based methods, the SR background density is learned by a normalising flow trained on the sidebands and interpolated into the SR. The ratio $w$ is therefore available \emph{as a by-product of the search itself}, as the sideband-to-SR density ratio evaluated on the resonant variable.

\begin{figure}[t]
  \centering
  \includegraphics[width=0.95\linewidth]{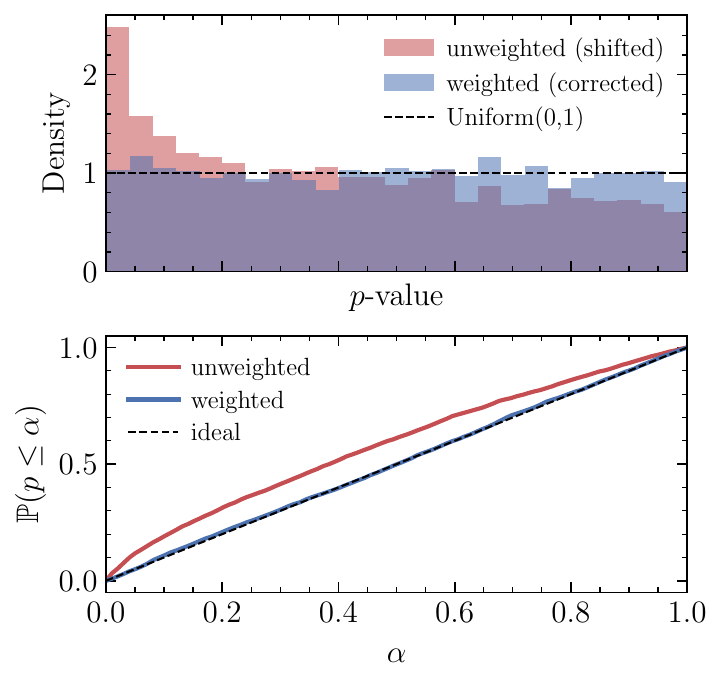}
  \caption{Weighted CP under covariate shift, no signal present. Upper: histograms of $\hat p$ (unweighted, piled at small values) and $\pw$ (weighted, uniform). Lower: exceedance curves. The red unweighted curve bows above the diagonal, the signature of invalidity; the blue weighted curve returns to the diagonal, restoring validity by reweighting.}
  \label{fig:weighted}
\end{figure}

Fig.~\ref{fig:weighted} is built on the mock model of Fig.~\ref{fig:calibration} with no signal injected. The calibration scores are drawn from the gamma background distribution $P$ of Sec.~\ref{sec:cp}, and the test (signal-region) background from $Q$, the same distribution rescaled by $25\%$, a controlled stand-in for the sideband-to-SR drift that the substructure--mass correlation produces on data. Because both distributions are known for this toy, the weight $w(s)=\mathrm{d}Q/\mathrm{d}P(s)$ is the exact likelihood ratio of the two gammas, and Eq.~\eqref{eq:wcp} reweights the calibration sample by it. The upper panel shows the resulting $p$-values, the unweighted $\hat p$ of Eq.~\eqref{eq:cp} (red), which piles up at small values because the drift makes background test scores look anomalous, against the weighted $\pw$ of Eq.~\eqref{eq:wcp} (blue), which is flat. The lower panel shows the corresponding exceedance curves. The unweighted exceedance breaches the bound ($\Prob(\hat p\le0.05)\approx0.10$, double the nominal), while the weighted version sits on the diagonal ($\approx0.045$). To probe robustness to an imperfect weight, we repeat the experiment with the exact ratio perturbed by multiplicative log-normal noise up to $\sigma=0.5$, and the weighted exceedance stays within sampling error of nominal, confirming the bounded-degradation prediction of Ref.~\cite{Barber2023shift}.

\subsection{Mondrian calibration}
\label{sec:mondrian}

Under exchangeability, Theorem~\ref{thm:cp} guarantees that the sample satisfies marginal coverage, that is, across the entire sample. However, this leaves shape-dependent observables rather exposed, since coverage is not guaranteed locally. Weighting corrects a smooth global shift. A complementary failure is \emph{heterogeneity}. The background score varies from bin to bin along the resonant variable, e.g. dijet mass. A single pooled calibration set can then be valid \emph{marginally} yet badly miscalibrated \emph{within} each bin, exactly the regime that produces localised bumps. Mondrian (group-conditional) CP~\cite{Vovk2013mondrian} calibrates separately within each bin.

\begin{definition}[Mondrian conformal $p$-value]
Partition the data into taxonomy cells (e.g. $m_{jj}$ bins) indexed by $b(x)$. For a test point in cell $b$ with calibration scores $\{s_i : b(x_i)=b\}$ of size $n_b$, define $\hat p_b(s)$ by Eq.~\eqref{eq:cp} using only that cell.
\end{definition}

\begin{proposition}[Bin-conditional validity; \cite{Vovk2013mondrian}]
\label{prop:mondrian}
If the test point is exchangeable with the calibration scores \emph{within its own cell} $b$, then $\Prob(\hat p_b(s)\le\alpha \mid b)\le\alpha$ for every cell, and hence also marginally.
\end{proposition}
\begin{proof}
Apply Theorem~\ref{thm:cp} conditionally on the cell, using only the within-cell calibration scores; the within-cell exchangeability hypothesis is weaker than global exchangeability. Marginal validity is followed by averaging the conditional bound over cells. This is the conditional-validity result of ref.~\cite{Vovk2013mondrian}.
\end{proof}

Two conditions are needed for the bound to mean what it says. The cells $b(x)$ must be fixed independently of the test outcomes. A taxonomy chosen after inspecting the scores, for instance, binning to isolate an apparent excess, voids the conditional guarantee. In the resonant setting, the cells are intervals of $m_{jj}$ fixed before unblinding, which satisfies this. Second, conditioning costs calibration statistics. Each cell carries only its own $n_b$ calibration scores, the resolution floor in a cell is $1/(n_b+1)$, and very fine binning trades validity for variance by emptying the cells. Bin widths must therefore be set by the smallest $p$-value a cell must resolve.

\begin{figure}[th]
  \centering
  \includegraphics[width=\linewidth]{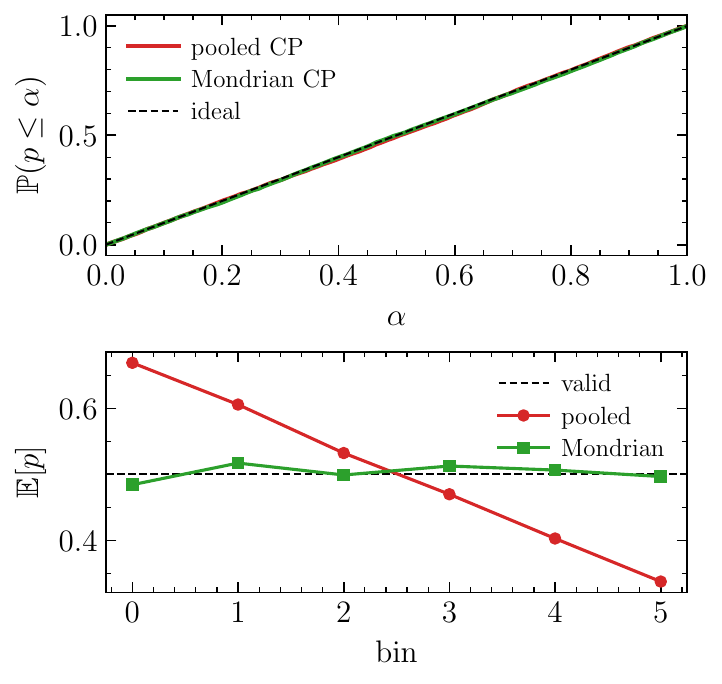}
  \caption{Mondrian vs pooled calibration across six resonant-variable bins. Upper: marginal exceedance (both look valid). Lower: per-bin mean $p$-value; the pooled estimator (red) drifts far from $0.5$ bin-by-bin while Mondrian (green) stays at $0.5$.}
  \label{fig:mondrian}
\end{figure}

Fig.~\ref{fig:mondrian} illustrates the failure of coverage in the binned sample. The pooled calibration is constructed across six bins, with the background level rising with bin index. The top panel shows that the \emph{marginal} exceedance of the pooled estimator lies on the diagonal, so a naïve marginal check would pronounce it valid. Bottom: the \emph{per-bin} mean $p$-value of the pooled estimator ranges from $0.66$ in the lowest bin to $0.33$ in the highest, indicating gross conditional miscalibration, whereas the Mondrian estimator holds at $0.5$ in every bin. 

In a previous study, we also demonstrated the failure of local coverage across classification, anomaly detection, and generative modelling for HEP applications, and noted that the coverage guarantee does not hold locally for a given distribution~\cite{Araz:2025vuw}. Although local coverage is not necessary for judging the quality of the overall sample, shape-dependent observables, such as energy-based observables in effective field theories, will require a more local coverage guarantee to preserve sensitivity in the tail of the distribution. Hence, it is essential to apply Mondrian calibration for analyses that require local sensitivity. Marginal validity can mask conditional invalidity, and it is the conditional failure that produces spurious bumps in the distribution.

\subsection{Robust calibration under a contaminated reference}
\label{sec:robust}

Unfortunately, designing a control region is not trivial. In many cases, the signal can leak into the control or validation regions, which are already present in numerous searches even without machine-learning intervention. This has been rigorously documented in LHC reinterpretation literature (e.g. see refs.~\cite{Evans:2013jna, Blanke:2015ulx, Altakach:2024jwk, Alguero:2022gwm}). Hence, leakage has to be accounted for. Let the calibration sample contain a fraction $\epsilon$ of outliers (signal), which might inflate the Type-I error rate. The opposite is true in the non-adversarial regime.

\begin{theorem}[Conservative validity under contamination; \cite{Bashari2025}]
\label{thm:robust}
Let the calibration set be drawn from the mixture $(1-\epsilon)P_0+\epsilon P_1$, with $P_1$ stochastically larger than $P_0$ in the anomaly score (i.e., the signal is more anomalous). Then, for a background test point, the conformal $p$-value of Eq.~\eqref{eq:cp} satisfies $\Prob(\hat p(s)\le\alpha)\le\alpha$; the procedure remains valid and is, in general, \emph{conservative}.
\end{theorem}
\begin{proof}[Proof idea]
The argument is a stochastic-dominance coupling. Fix a background test score $s$. Let $C_0=(s_1^0,\dots,s_n^0)$ be a clean calibration sample drawn i.i.d.\ from $P_0$ and $C=(s_1,\dots,s_n)$ the contaminated sample from $(1-\epsilon)P_0+\epsilon P_1$. First-order stochastic dominance $P_1\succeq P_0$ in the score means there is a coupling under which each contaminated draw is pointwise at least as large as its clean counterpart, $s_i\ge s_i^0$ almost surely. Under that coupling the upper-tail count cannot decrease, $|\{i:s_i\ge s\}|\ge|\{i:s_i^0\ge s\}|$, so by Eq.~\eqref{eq:cp} the contaminated $p$-value satisfies $\hat p(s)\ge\hat p_0(s)$ pointwise, hence $\hat p$ is stochastically larger than $\hat p_0$. Since the clean-reference $p$-value is super-uniform (Theorem~\ref{thm:cp}), $\Prob(\hat p(s)\le\alpha)\le\Prob(\hat p_0(s)\le\alpha)\le\alpha$, with the first inequality typically strict, so the procedure is conservative. The precise non-adversarial conditions and the accompanying power loss are quantified in ref.~\cite{Bashari2025}.
\end{proof}

The conservativeness depends on the stochastic-dominance premise and is not unconditional. If the contamination is adversarial, or if the leaked outliers populate \emph{low} scores so that $P_1\not\succeq P_0$, or if the contamination enters the score \emph{training} rather than only the calibration reference, the coupling above fails and the type-I rate can exceed $\alpha$. The result therefore certifies fail-safe behaviour only for the non-adversarial, high-score contamination characteristic of a resonance localised in $m_{jj}$, which leaks signal-like (high-score) events into the reference. We verify this regime empirically in Sec.~\ref{sec:lhco_power}.

Fig.~\ref{fig:contam} shows the effect of contamination. As the reference contamination grows from $0\%$ (blue) to $10\%$ (red), the exceedance curve on background test points sinks progressively \emph{below} the diagonal. At $5\%$ (green) contamination $\Prob(\hat p\le0.05)\approx0.009$, far under nominal.
\begin{figure}[t]
  \centering
  \includegraphics[width=0.95\linewidth]{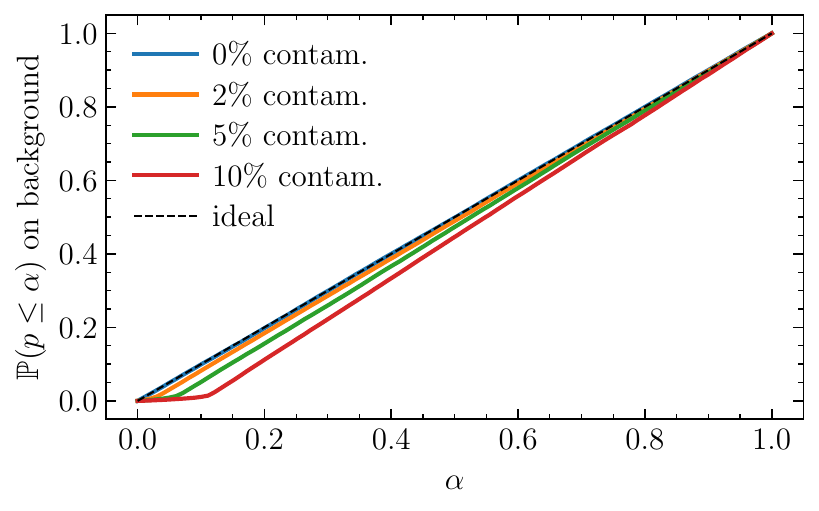}
  \caption{Effect of signal contamination of the reference set on background $p$-values.}
  \label{fig:contam}
\end{figure}
Signal leakage into the sideband costs sensitivity (power) but does \emph{not} invalidate the error control. The conformal layer fails safe. This is the sense in which finite-sample conformal validity earns its keep where the asymptotic local $p$-value gives no such guarantee, and it is why a modest, bounded power loss, rather than a silent false discovery, is the price of contamination. An active cleaning step can recover power if a labelling budget is available~\cite{Bashari2025}.

\subsection{The joint conformal \& Gross--Vitells layer}
\label{sec:gv}

We now combine the finite-sample-calibrated local $p$-values with a look-elsewhere correction. The construction treats the conformal $p$-value field across the resonant variable as the object to which the standard trials-factor theory is applied. We stress at the outset that the finite-sample conformal guarantee attaches to the local $p$-value construction of Sec.~\ref{sec:cp}. The global trials-factor step below is asymptotic, calibrated on background-only toys, and is \emph{not} itself a finite-sample-valid object.

\subsubsection{From $p$-field to a local-significance field}

Scan the resonant variable in windows $m~\in~\{1,\dots,M\}$. In window $m$, let the events carry conformal $p$-values (Mondrian, weighted as needed). Summarise the window by the count of conformally anomalous events $k_m=|\{p\le\alpha\}|$ out of $n_m$, and standardise against the background expectation $\alpha n_m$,
\begin{equation}
  Z(m) \;=\; \frac{k_m - \alpha n_m}{\sqrt{n_m\,\alpha(1-\alpha)}}\ .
  \label{eq:Zfield}
\end{equation}
Under background-only and valid per-window calibration, $\E[k_m]=\alpha n_m$ by Prop.~\ref{prop:unif}, so $Z(m)$ is mean-zero.\footnote{See also~\cite{Donoho_2004}. For the unfamiliar reader, $Z(m)$ is the good-old (observed $-$ expected) $/$ standard deviation, and the denominator is the standard deviation of the binomial distribution.} Its variance is \emph{not} the binomial $n_m\alpha(1-\alpha)$. The within-window conformal $p$-values share one calibration set and are positively dependent, so 
\begin{eqnarray}
   \mathrm{Var}(k_m)=c_m\,n_m\alpha(1-\alpha)
   \label{eq:var}
\end{eqnarray}
with an over-dispersion $c_m\ge1$ that we measure and correct in Sec.~\ref{sec:nullvar}.\footnote{Conditioning on the calibration set $C$, the within-window indicators $\mathbf 1\{\hat p_i\le\alpha\}$ are i.i.d.\ Bernoulli with rate $q_C=\Prob(\hat p\le\alpha\mid C)$, so by the law of total variance $\mathrm{Var}(k_m)=n_m\alpha(1-\alpha)+n_m(n_m-1)\,\mathrm{Var}(q_C)$ and $c_m=1+(n_m-1)\,\mathrm{Var}(q_C)/[\alpha(1-\alpha)]$, the design effect of the equicorrelation $\rho=\mathrm{Var}(q_C)/[\alpha(1-\alpha)]$ induced by the shared calibration set. For split conformal, $q_C$ is the coverage at the calibration $(1-\alpha)$-quantile; the standard variance of that empirical quantile gives $\mathrm{Var}(q_C)\approx\alpha(1-\alpha)/n_{\rm cal}$ under exchangeability, hence $c_m\approx1+n_m/n_{\rm cal}$, while a covariate shift rescales $\mathrm{Var}(q_C)$, which is why we measure $c_m$ on data in Sec.~\ref{sec:nullvar}.} Standardised by that variance, $Z(m)$ is approximately unit-variance. Smooth correlations between neighbouring windows make it a smooth random field, the discrete analogue of the $\chi^2$ field of ref.~\cite{Gross2010}.

\subsubsection{Gross--Vitells trials factor on the conformal field}

The global $p$-value of the scan maximum is bounded by the Gross--Vitells up-crossing formula~\cite{Gross2010, Vitells2011},
\begin{eqnarray}
  \Prob\!\Bigl(\max_m Z(m)\ge z\Bigr)&\nonumber \\ \le\;\Prob&\bigl(Z_{\mathrm{loc}}\ge z\bigr)
   + \E[N_{u}]\, e^{-(z^2-u^2)/2},
  \label{eq:gv}
\end{eqnarray}
where $\E[N_{u}]$ is the mean number of up-crossings of a reference level $u$ by $Z(m)$, estimated from a small set of background-only toys, and $\Prob(Z_{\mathrm{loc}}\ge z)$ is the local tail. The exponential factor is the trials factor. It converts a local significance into a global one using only the field's up-crossing rate, exactly as in the asymptotic profile-likelihood case, but here driven by the \emph{conformal} field in Eq.~\eqref{eq:Zfield}.

We state the global step as an asymptotic, toy-calibrated procedure rather than a finite-sample theorem. Per-window conformal validity controls the first two moments of the count field but does not, by itself, supply the smooth Gaussian random-field structure that the Gross--Vitells up-crossing theory requires. First and second moments are not sufficient for the random-field assumptions of ref.~\cite{Gross2010}. We therefore record the construction as the following procedure.

\begin{quote}
\emph{Procedure (asymptotic global $p$-value).}
Standardise the conformal count field $Z(m)$ using an independent background-only null calibration, so that under the background hypothesis it is mean-zero and approximately unit-variance with the over-dispersion of Eq.~\eqref{eq:var} absorbed (Sec.~\ref{sec:nullvar}). If the resulting scan field is empirically well approximated by the smooth random-field assumptions of ref.~\cite{Gross2010}, then Eq.~\eqref{eq:gv} gives an asymptotic trials-factor approximation, or upper bound, to the global false-positive probability of the scan. The finite-sample conformal guarantee applies to the local $p$-value construction (Theorem~\ref{thm:cp}, Propositions~\ref{prop:unif},~\ref{prop:mondrian}, Theorems~\ref{thm:wcp}/\ref{thm:robust}), not to the Gross--Vitells global $p$-value itself.
\end{quote}

In practice, we do not rely on the random-field approximation alone. We corroborate Eq.~\eqref{eq:gv} against the empirical global tail obtained from background-only toys (Fig.~\ref{fig:gv} below and on data in Sec.~\ref{sec:lhco_scan}), and quote the Gross--Vitells value only as an upper bound whose Gaussian-field assumptions are weakest in the extreme tail. A fully empirical alternative, calibrating the global null directly from independent background pseudoexperiments throughout the entire pipeline, is presented as a deferred study in the scope-and-limits discussion of Sec.~\ref{sec:limits}.

\begin{figure}[t]
  \centering
  \includegraphics[width=\linewidth]{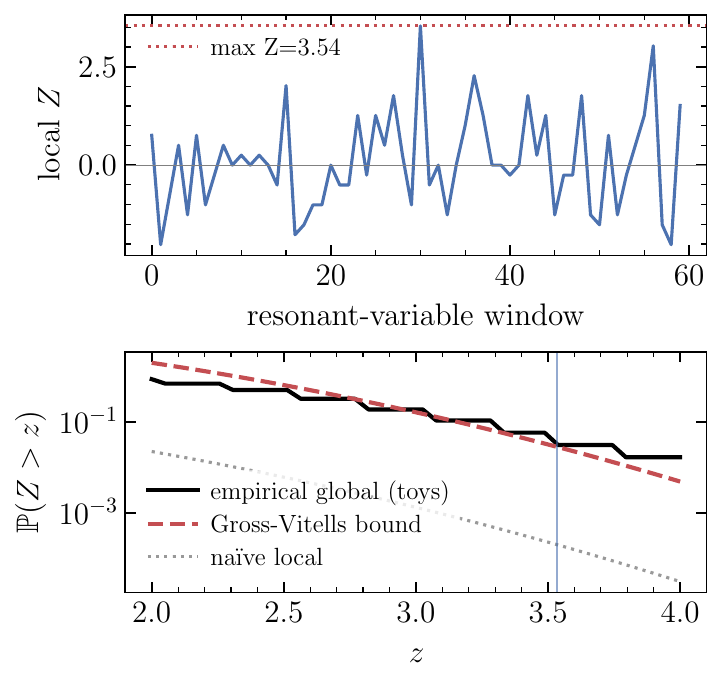}
  \caption{Joint conformal \& Gross--Vitells layer. Upper: conformal local-significance field $Z(m)$ of Eq.~\eqref{eq:Zfield} with an injected signal at window $30$. Lower: global vs local significance; the GV-bound Eq.~\eqref{eq:gv} (red dashed) matches the empirical global tail from toys (black) and lies far above the naïve local tail (grey). The reportable significance is the GV global value; the gap to the naïve local value is the trials factor ($\approx125$ here).}
  \label{fig:gv}
\end{figure}

Fig.~\ref{fig:gv} shows the pipeline end-to-end on a $60$-window scan with a moderate signal injected into one window. The upper panel shows the conformal local-significance field $Z(m)$, peaking at the injected window with $Z_{\max}=3.28$. The lower panel shows three curves on a log scale: the naïve local tail $\Prob(Z_{\mathrm{loc}}\ge z)$ (dotted black), the Gross--Vitells global bound Eq.~\eqref{eq:gv} calibrated with $\E[N_u]\approx12$ at $u=0.5$ (dashed red), and the empirical global tail from $2000$ background toys (solid black). The GV bound tracks the empirical global tail in the relevant range ($z\gtrsim2.5$), while the naïve local curve lies two orders of magnitude below. For the injected excess, a local $p\approx5\times10^{-4}$ ($Z\approx3.3$, naïvely ``evidence'') becomes a global $p\approx0.06$ once the scan over $60$ correlated windows is accounted for, a trials factor of $\approx 125$. This is the single most important output of the method; an honest global significance, neither the over-optimistic local value nor an uncontrolled shortlist.

\subsection{What does finite-sample validity add?}
\label{sec:disc}

Given that the Wilks theorem is accurate to a fraction of a $\sigma$ at LHC luminosities, finite-sample validity might seem like a marginal gain over the asymptotic local $p$-value. It can be argued that the real bottleneck is background mismodelling and systematics. However, the conformal layer's contribution is \emph{not} a sharper local $p$-value. Rather, the entire pipeline, from score to global significance, remains valid when its assumptions are stated and checked, and degrades \emph{gracefully and conservatively} when they are imperfect (Theorems~\ref{thm:wcp},~\ref{thm:robust}). The asymptotic local $p$-value has no analogous guarantee against background mismodelling. A biased background template inflates the local significance silently, whereas a mis-specified conformal calibration shows up as non-uniform background $p$-values (Prop.~\ref{prop:unif}) and, where it is a shift or a contamination, is either corrected (Section~\ref{sec:weighted}) or fails safe (Section~\ref{sec:robust}).\footnote{A ``background template'' is our model of what the background looks like, used as the null hypothesis. In a bump hunt, it's the expected count (or shape) in each bin under ``no signal.'' The asymptotic $p$-value asks: given my observed count $N$ and my assumed expected background $b_{\rm m}$, how surprising is $N$? A biased template means $b_{\rm m}$ is simply wrong. It doesn't match the true background $b_{\rm true}$ for reasons that have nothing to do with the signal. Causes include a mis-tuned Monte Carlo generator, the wrong parton distribution functions, detector mismodelling, an oversimplified fit function for the background shape, or an incorrect normalisation.} Fig.~\ref{fig:asym} makes this asymmetry explicit in a single-bin counting experiment with no signal present. The asymptotic local significance (Asimov significance) for an observed count $N$ against an assumed mean $b_{\rm m}$ is
\begin{equation}
    Z_{\rm asym} = \mathrm{sign}(N-b_{\rm m})\sqrt{2\left[N\ln(N/b_{\rm m}) - (N-b_{\rm m})\right]},
    \label{eq:asimov}
\end{equation}
while the conformal $p$-value calibrates $N$ against a reference set of background-only counts. We construct the toy as follows. A score cut is placed at the $90$th percentile of the (single-region) background score, and a ``window'' is a subsample of $n_{\rm win}=2000$ background events. The number passing the cut is $N$, with true background expectation $b_{\rm true}=n_{\rm win}\,\Prob(s\ge s_{\rm cut})$. We draw $4000$ such background-only windows and, separately, a large reference dataset of $2\times10^{4}$ background-only counts that serves as the conformal calibration set. The asymptotic test scores each window by Eq.~\eqref{eq:asimov} against an \emph{assumed} mean $b_{\rm m}=b_{\rm true}(1+\beta)$, where $\beta$ is a deliberate template bias, and reads off $p=\Prob(Z_{\rm asym}\ge z)$; the conformal test scores the same window by Eq.~\eqref{eq:cp} against the reference counts, which never consult $b_{\rm m}$. 

\begin{figure}[th]
  \centering
  \subfigure[False-positive rate at $\alpha=0.05$ on pure background vs template bias: the asymptotic rate (red) diverges as the template is underestimated, and the conformal rate (blue) stays nominal.\label{fig:asym-fpr}]{\includegraphics[width=0.95\linewidth]{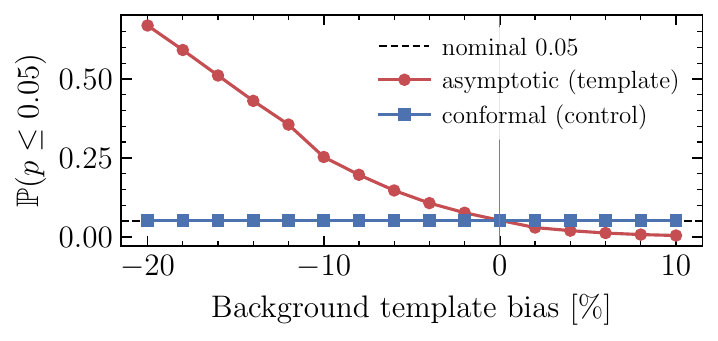}}
  \subfigure[$p$-value histograms at $-10\%$ bias: asymptotic piles up near zero (fake significance), conformal stays uniform.\label{fig:asym-pvals}]{\includegraphics[width=0.95\linewidth]{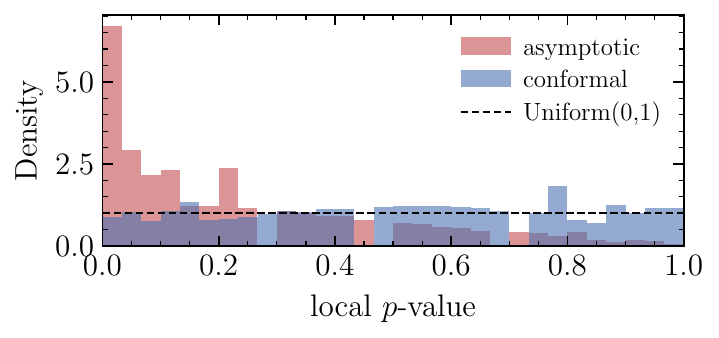}}
  \subfigure[The conformal self-diagnostic on true background: a biased control sample is exposed as a non-uniform exceedance curve, whereas a good control sits on the diagonal.\label{fig:asym-selfdiag}]{\includegraphics[width=0.95\linewidth]{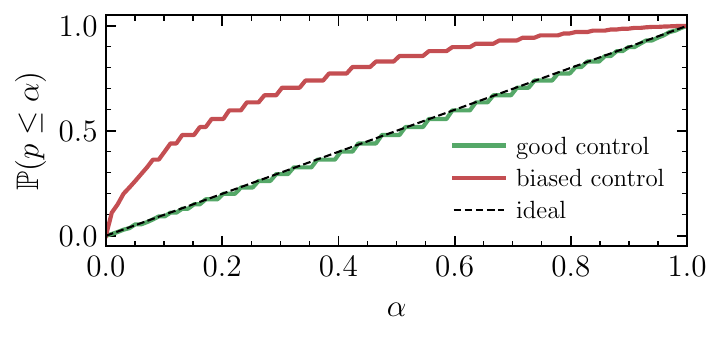}}
  \caption{Asymptotic vs conformal local $p$-values under background mismodelling, pure background, no signal. The dangerous, silent failure is the red curve in (a) on the underestimate side; conformal converts that hidden failure into the checkable non-uniformity of (c).}
  \label{fig:asym}
\end{figure}

Fig.~\ref{fig:asym-fpr} shows the background-only false-positive rate at $\alpha=0.05$ plotted against the template bias $\beta\in[-20\%,+10\%]$. The asymptotic rate (red) is correct only at $\beta=0$ and diverges as the template is underestimated. A $10\%$ underestimate turns roughly one in four background windows into a nominal $5\%$ ``excess''. The conformal rate (blue) remains flat at the nominal $0.05$ over the entire range because it calibrates against the background reference rather than a template. Fig.~\ref{fig:asym-pvals} shows the per-window local $p$-values with $-10\%$ bias. The asymptotic distribution (red) piles up near zero, the signature of manufactured significance, whereas the conformal distribution (blue) is uniform, as a valid $p$-value must be on background. Fig.~\ref{fig:asym-selfdiag} is the one honest caveat. If the conformal \emph{control} itself is biased, the conformal $p$-values inflate too. We plot the exceedance $\Prob(\hat p\le a)$ on a trusted background sample for a good and for a deliberately biased control. The biased control bows off the diagonal, exactly the non-uniformity of Prop.~\ref{prop:unif}, so the failure is \emph{visible}. The asymptotic method has no such self-check, since its only reference is the very template that is wrong.

Second, the recent demonstration that AD $p$-values are badly miscalibrated when training and testing share data~\cite{Hein2025} is a finite-sample, data-reuse pathology that no asymptotic formula addresses. The split/weighted/Mondrian conformal discipline is a direct response.

Third, the trials-factor layer of Sec.~\ref{sec:gv} is agnostic to whether the local significance is asymptotic or conformal. Our claim is only that pairing it with a conformal field yields a global significance that is simultaneously trials-factor-aware and equipped with the finite-sample robustness properties above. To our knowledge, no existing HEP pipeline delivers both at once.

\section{A worked example on LHC Olympics data}
\label{sec:lhco}

The preceding sections established each component on controlled toys, in which the ground truth is known and every assumption can be dialled in. We now demonstrate the full pipeline on the public collider dataset, the LHC Olympics 2020 R\&D dataset~\cite{Kasieczka2021}, using the resonant signal-region/sideband construction of CWoLa hunting and CATHODE~\cite{Collins2019cwola, Hallin2021cathode}. The goal is not to set a new sensitivity record but to show, on data with all the pathologies the toys idealised, that the conformal layer (i) exposes a calibration failure that the standard pipeline would not flag, (ii) repairs it with the label-free weighted construction of Section~\ref{sec:weighted}, and (iii) converts an uncalibrated score into an honest statement about the presence or absence of a localised excess.

\subsection{Dataset and preprocessing}
\label{sec:lhco_data}

The LHC Olympics R\&D dataset~\cite{Kasieczka2021} consists of simulated proton--proton collisions at $\sqrt{s}=13$~TeV. The background is generic QCD dijet production. The injected signal is a resonant cascade $Z'\to X\,Y$ with $X\to q\bar q$ and $Y\to q\bar q$, at masses $(m_{Z'},m_X,m_Y)=(3.5,\,0.5,\,0.1)$~TeV, so the signal is localised as a bump in the dijet invariant mass $m_{jj}$ near $3.5$~TeV. Each event is reduced to a small set of high-level jet features defined on the two leading large-radius jets, ordered by mass so that $m_{j_1}<m_{j_2}$:
\begin{equation}
  x = \bigl(m_{j_1},\; \Delta m_j \equiv m_{j_2}-m_{j_1},\;
            \tau_{21}^{j_1},\; \tau_{21}^{j_2}\bigr),
  \label{eq:lhco_features}
\end{equation}
where $\tau_{21}=\tau_2/\tau_1$ is the $n$-subjettiness ratio that discriminates two-pronged (signal-like) from one-pronged (QCD-like) jets. The resonant variable $m_{jj}$ is held separate from $x$ and used only to define regions, never as a classifier input, so that the anomaly score does not sculpt the $m_{jj}$ spectrum.

\paragraph{Signal region and sidebands.}
The signal region (SR) is the mass window $m_{jj}\in[3.3,3.7]$~TeV, chosen to bracket the $3.5$~TeV resonance. The sidebands (SB) are the complementary region in $m_{jj}$, signal-depleted and treated as a control sample of (nearly pure) background. Two families of resonant anomaly detection build a score on this geometry. CWoLa hunting~\cite{Metodiev:2017vrx, Collins2018, Collins2019cwola} trains a classifier to separate SR events from SB events directly, on features chosen to be (ideally) uncorrelated with $m_{jj}$, exploiting that the SR is a mixture with a larger signal fraction than the signal-depleted SB. CATHODE~\cite{Hallin2021cathode}, and with explicit likelihoods ANODE~\cite{Nachman2020}, instead trains a conditional density estimator on the SB, interpolates it into the SR, and \emph{samples} synthetic background events at the SR masses, then trains a classifier to separate the SR data from those samples. Sampling the background at the SR mass is what makes CATHODE robust to feature--mass correlations~\cite{Hallin2021cathode}. We adopt the relatively simple CWoLa-hunting realisation. Using the sideband events themselves as the background class leaves the SB$\to$SR background shift in place, which is the very exchangeability failure the conformal layer is designed to expose and repair.

For the input features, we use the preprocessed files distributed with the public CATHODE repository,\footnote{\url{https://github.com/HEPML-AnomalyDetection/CATHODE}}, which provide Eq.~\eqref{eq:lhco_features} already split into ``inner'' (SR) and ``outer'' (SB) samples and standardised in the CATHODE convention. Each feature is shifted and scaled to zero mean and unit variance using the SB statistics, and masses are expressed in TeV. The splits are an SR training pool of $6.0\times10^{4}$ events ($\sim$$0.6\%$ injected signal); an independent held-out SR evaluation pool of $6.2\times10^{4}$ events ($377$ signal events), from which the conformal calibration and test points of Sec.~\ref{sec:lhco_validity} are drawn; an SB reservoir of $3.8\times10^{5}$ nearly background-only events that serves as the conformal calibration source; and a pure-signal pool for the injection studies of Sec.~\ref{sec:lhco_power}.

\paragraph{Anomaly score.}
The score is the output of a binary classifier trained, in the CWoLa-hunting manner, to separate the SR pool (class $1$) from a background reference of $1.5\times10^{5}$ events subsampled at a fixed seed from the SB training pool (class $0$), using the four features of Eq.~\eqref{eq:lhco_features} only. $m_{jj}$ is never an input, so the score cannot sculpt the resonant spectrum directly. In the CWoLa limit, the Bayes-optimal classifier for this region-label problem is a monotonic function of the true signal-to-background likelihood ratio~\cite{Metodiev:2017vrx}, so separating SR from SB isolates the signal-like population with no per-event truth labels. We use a gradient-boosted decision tree rather than a deep network for a simple implementation; it is fast and essentially free of training stochasticity, easy-peasy.\footnote{scikit-learn (version 1.6.1)~\cite{scikit-learn} \texttt{HistGradientBoostingClassifier} with maximum iteration of 300, $0.08$ learning rate, max depth at $4$, L2 regularisation at $1$ and a $20\%$ early-stopping validation fraction.} The per-event score is $s=\hat p(\mathrm{class}=1\mid x)$, oriented so that larger is more SR-like, hence more signal-like. On the held-out SR pool (truth labels used only for this check), it reaches an area under the ROC curve of $0.875$ and a peak significance-improvement characteristic $\mathrm{SIC}\ (\epsilon_S/\sqrt{\epsilon_B})\approx1.95$, consistent with published weakly-supervised performance at this injection~\cite{Collins2019cwola, Hallin2021cathode}. The full CATHODE flow reaches a substantially higher SIC. On this benchmark, Refs.~\cite{Hallin2021cathode, Golling:2023yjq} report maximal SIC values of $\approx14$ for CATHODE, $\approx11$ for CWoLa hunting, and $\approx6.5$ for ANODE, against our $\approx2$, so the tree is deliberately and markedly weaker. Because the conformal layer acts on the score \emph{after} the classifier, it is agnostic to this choice, and a stronger classifier propagates through the identical machinery. We thus cleanly separate two questions: the classifier sets the raw sensitivity, while the conformal layer sets the calibration and the honesty of the final significance.

\subsection{Quantifying the uncertainty on an exceedance}
\label{sec:lhco_size}

Every validity number below is an exceedance estimate $\Prob(\hat p\le\alpha)$, a proportion measured over a finite test sample against a single calibration draw, and we quote it with a \emph{total} error from two independent sources combined in quadrature. The first is the binomial (test) component,
\begin{equation}
  \mathrm{SE}_{\rm bin} = \sqrt{\hat P(1-\hat P)/n_{\rm test}},
  \label{eq:se_bin}
\end{equation}
where $n_{\rm test}$ is the number of test events; this is the leading term and uses the \emph{test} size, not the calibration size. The second is a calibration-draw component. Since all $p$-values in one experiment share a single calibration set, they are positively correlated, and resampling the calibration set (a bootstrap) and recomputing the exceedance estimates contributes $\mathrm{SE}_{\rm cal}$. The total error is $\mathrm{SE}_{\rm tot}=\sqrt{\mathrm{SE}_{\rm bin}^2+\mathrm{SE}_{\rm cal}^2}$. For the validity measurement that follows, $n_{\rm test}\approx6\times10^{4}$ gives $\mathrm{SE}_{\rm bin}\approx0.001$, 
and the calibration component is of the same order when the full SB reservoir is used, so the total error is $\approx0.0014$, which we round up to $\approx0.002$ when quoting results.

\subsection{Restoring exchangeability on data}
\label{sec:lhco_validity}

We first apply the validity diagnostic of Prop.~\ref{prop:unif} to real data. Fig.~\ref{fig:lhco_validity-score} shows the anomaly score distribution for the background (blue) and signal (red) within the signal region. Calibrating the SR scores against the SB reservoir and examining the conformal $p$-values of \emph{signal-region background} events, these should be uniform if the SB and SR backgrounds are exchangeable. They are not. Fig.~\ref{fig:lhco_validity-exchange} shows the exceedance curve for the sideband-calibrated background (red) bowing clearly above the diagonal. We measure 
\begin{equation}
  \Prob(\hat p \le 0.05) = 0.087 \pm 0.002 \ \text{(sideband-calibrated)},
\end{equation}
against the nominal $0.05$, a $\sim$$18\sigma$ departure given the total error defined above (Sec.~\ref{sec:lhco_size}). For comparison, calibrating SR background against a held-out half of SR background (green) (exchangeable by construction, the ``ideal'' control that is unavailable in a real blinded search) restores $\Prob(\hat p\le0.05)=0.045\pm0.002$, on the diagonal. The non-uniformity is therefore not a statistical artefact but a genuine breakdown of exchangeability. The jet substructure features are correlated with $m_{jj}$, so the background score distribution drifts between SB and SR (the mean score shifts from $0.283$ in the SB to $0.292$ in the SR). We use the idealised-AD/CWoLa score deliberately as a stress test. By training directly on the SR-versus-SB difference, it is the most mass-correlated choice, so the shift it produces sits at the large end of what one should expect. A more mass-robust detector, such as the full CATHODE flow whose background sample is drawn at the SR mass, would generally exhibit a smaller shift; but its size is not knowable a priori and must be measured, which is what the diagnostic provides, and whatever its size, it can be remedied by the weighted correction.

\begin{figure}
    \centering
    \subfigure[Anomaly-score distributions for SR background and (truth-labelled) SR signal; the signal populates higher scores, as a working detector requires.\label{fig:lhco_validity-score}]{\includegraphics[width=0.95\linewidth]{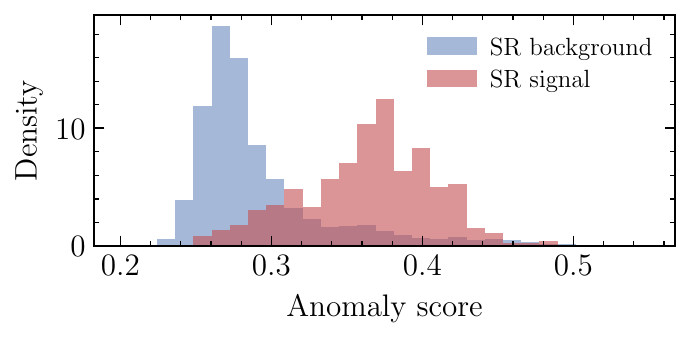}}

    \subfigure[Validity diagnostic: the conformal exceedance $\Prob(\hat p\le a)$ for SR \emph{background} events. The sideband-calibrated curve (red) bows above the diagonal, a real exchangeability failure, while the SR-self-calibrated curve (green, the unavailable ideal) lies on the diagonal.\label{fig:lhco_validity-exchange}]{\includegraphics[width=0.95\linewidth]{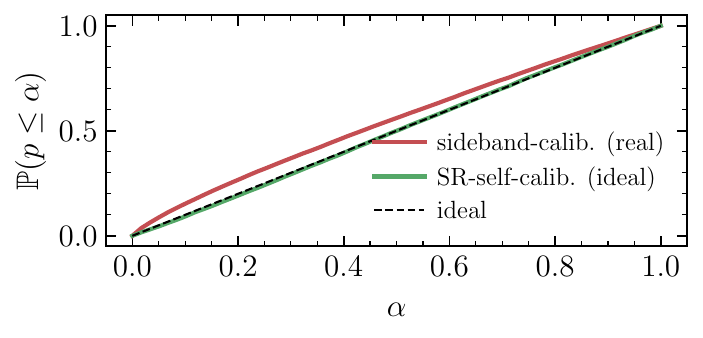}}

    \subfigure[The label-free weighted correction of Sec.~\ref{sec:weighted} pulls the exceedance back onto the diagonal (blue). Shaded bands are $\pm1$ total standard error (binomial test component combined in quadrature with the calibration-bootstrap component; Sec.~\ref{sec:lhco_size}).\label{fig:lhco_validity-weighted}]{\includegraphics[width=0.95\linewidth]{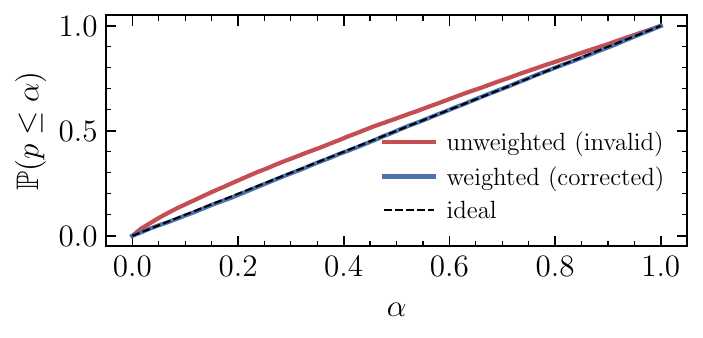}}

    \caption{Conformal calibration on LHC Olympics R\&D data.}\label{fig:lhco_validity}
\end{figure}

Crucially, this failure is invisible to the standard pipeline. An asymptotic local $p$-value computed against an SB-derived background template would inherit the same $m_{jj}$--feature correlation as a silent normalisation bias and report inflated significance with no internal warning. The conformal diagnostic, by contrast, \emph{announces} the problem as a visibly non-uniform background exceedance.

We then apply the label-free weighted correction of Section~\ref{sec:weighted}. We estimate the SR-to-SB density ratio $w(s)=p_{\rm SR}(s)/p_{\rm SB}(s)$~\cite{Andreassen:2020nkr} by training a classifier to separate \emph{observed} SR events (signal and background mixed, as in a real search) from SB events on the score alone, and from the weighted conformal $p$-value shown in Eq.~\eqref{eq:wcp}. No truth labels enter the weight; the SR and SB are defined purely by $m_{jj}$. Fig.~\ref{fig:lhco_validity-weighted} shows the result. The weighted background exceedance returns to the diagonal,
\begin{equation}
  \Prob(\pw \le 0.05) = 0.050 \pm 0.002\ \text{(label-free weighted)},
\end{equation}
restoring validity. The exchangeability failure of a real resonant anomaly-detection sideband is thus not merely detected but corrected, using only the kind of density ratio the search already estimates.

\subsection{The power cost of the weighted correction}
\label{sec:lhco_power}

The weighting that restores background validity is learned from the \emph{observed} signal region, which contains any signal present, so one must ask whether it also suppresses a genuine excess. It does, by a measurable amount, which is quantified in Fig.~\ref{fig:power-sig}. We set up a controlled injection. An observed SR sample of $n=4\times10^{4}$ events is assembled from SR-background events and a fraction $f$ of events drawn from the pure-signal pool, mimicking a real, label-unknown SR of signal strength $f$. From this observed SR, we learn the label-free density ratio $w$ exactly as in Sec.~\ref{sec:lhco_validity}, by training a classifier to separate the observed SR from the SB on the score alone. We then draw an \emph{independent} set of signal events and compute their conformal $p$-values two ways, unweighted (sideband-calibrated, Eq.~\eqref{eq:cp}) and label-free weighted (Eq.~\eqref{eq:wcp}), and record the \emph{median} signal $p$-value as a scalar summary of how strongly the detector flags the signal. Fig.~\ref{fig:power-sig} plots this median against the injected fraction $f$ for the two estimators. A genuine signal should produce \emph{small} $p$-values, so the lower curve is the more sensitive one, and the weighted curve lying above the unweighted one is the power cost made visible. Scanning $f\in\{1,2,5\}\%$, the weighted median rises by roughly a factor of two over the unweighted one (from $0.049$ to $0.087$ at $f=1\%$, and from $0.048$ to $0.105$ at $f=5\%$). The correction is therefore conservative on signal as well as on background. It buys validity with power. The detector is not insensitive, however. Even under the plain unweighted calibration, the signal median sits at $\approx0.05$, so about half of the injected signal events already fall below the $\alpha=0.05$ threshold. The weakly-supervised tree simply has no loose-cut counting reach at the native $0.6\%$ injection. The weighted $p$-value is thus a \emph{guarded} statistic. A discovery that survives it is robust to the sideband shift, whereas the unweighted statistic should be used only when exchangeability is independently established.

\begin{figure}[th]
    \centering
    \subfigure[Median conformal $p$-value of injected signal events versus injected fraction $f$, unweighted (sideband-calibrated) and label-free weighted; the weighting roughly doubles the signal $p$-value, the price of restoring validity.\label{fig:power-sig}]{\includegraphics[width=0.95\linewidth]{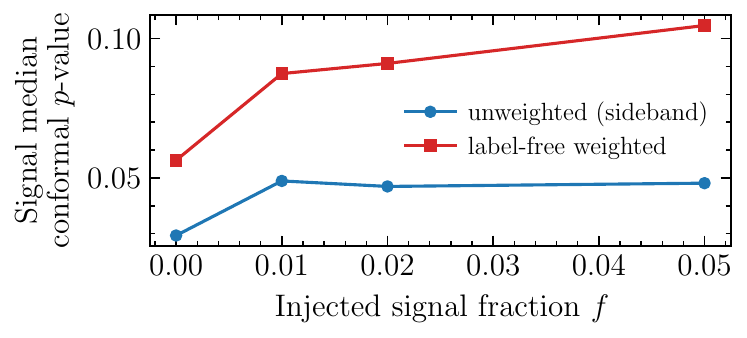}}
    
    \subfigure[Background false-positive rate $\Prob(\hat p\le0.05)$ for the label-free weighted $p$-values, versus the signal contamination $\epsilon$ of the observed signal region from which the weight is learned (blue, $\pm {\rm SE}_{\rm tot}$). Black dashed: nominal $0.05$; red dotted: the unweighted, biased baseline ($\approx0.087$). The weighted rate stays at or below nominal and grows more conservative with $\epsilon$.\label{fig:power-bkg}]{\includegraphics[width=0.95\linewidth]{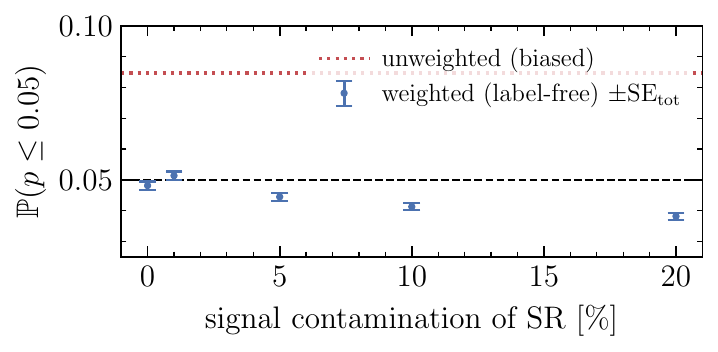}}

    \caption{Behaviour of the label-free weighted correction with signal contamination on LHC Olympics data.}
    \label{fig:power}
\end{figure}

A complementary question is what happens to \emph{background} validity when the weight is itself learned from a signal-contaminated signal region, since in a real search the observed SR always contains whatever signal is present (Fig.~\ref{fig:power-bkg}). For each contamination level $\epsilon\in\{0,1,5,10,20\}\%$ we build an observed SR with signal fraction $\epsilon$ (background from SR-background, signal from the pure-signal pool), learn the label-free ratio $w_\epsilon$ from that contaminated SR against the sideband, and measure the weighted background false-positive rate $\Prob(\pw\le0.05)$ on held-out SR-background events. The rate stays at or below the nominal $0.05$ at every $\epsilon$ and drifts progressively \emph{below} it as the contamination grows, while the unweighted baseline (red dotted) sits at the biased value of Sec.~\ref{sec:lhco_validity} ($\approx0.087$). The mechanism is the fail-safe direction of reference contamination (Theorem~\ref{thm:robust}). Signal populates high scores, so contamination only inflates $w_\epsilon$ in the high-score tail, making background events look \emph{less} anomalous rather than more. Signal contamination of the signal region therefore costs a little power but never background validity; the one premise that must still hold is that the \emph{sideband} is nearly signal-free, which a resonance localised in $m_{jj}$ guarantees. 
The Gross--Vitells global significance of the worked example is also stable in the free reference level $u_0$. Recomputing the bound Eq.~\eqref{eq:gv} on the toy-null field over $u_0\in\{0,0.3,0.5,0.8,1.0,1.5\}$, the mean up-crossing count falls from $\E[N_{u_0}]\approx8$ to $3$ as $u_0$ rises, but the resulting global significance stays at $0\sigma$ throughout (global $p$-value $\gtrsim0.8$, the bound saturating at unity), consistent with the null, so the global result does not hinge on that free choice.


\subsection{Template-free robustness on data}
\label{sec:lhco_asym}

Section~\ref{sec:disc} argued that the conformal score, calibrated against a control sample, is immune to background-template mismodelling in a way the asymptotic local $p$-value is not.\footnote{CATHODE's background dependence has been mitigated in a follow-up paper, LaCATHODE~\cite{Hallin:2022eoq}.} Similar to Fig.~\ref{fig:asym-fpr}, Fig.~\ref{fig:lhco_asym} demonstrates this on the LHCO scores. We define a single-bin counting experiment from the real scores (events above a fixed score cut, counted in the SR) and compare the asymptotic Poisson significance computed against an \emph{assumed} background expectation $b_{\rm m}$ with the conformal assessment calibrated against the control. As the assumed template is increasingly underestimated, the asymptotic false-positive rate diverges; a $10\%$ underestimate sends $\Prob(p\le0.05)$ from the nominal $0.05$ to $\approx0.42$ on pure background, manufacturing fake excesses, while the conformal rate stays flat at $\approx0.056$ because it never consults the template. This is the central practical advantage of the method, restated in terms of data. It fails safe against the dominant systematic in real searches.

\begin{figure}[th]
  \centering
  \includegraphics[width=0.95\linewidth]{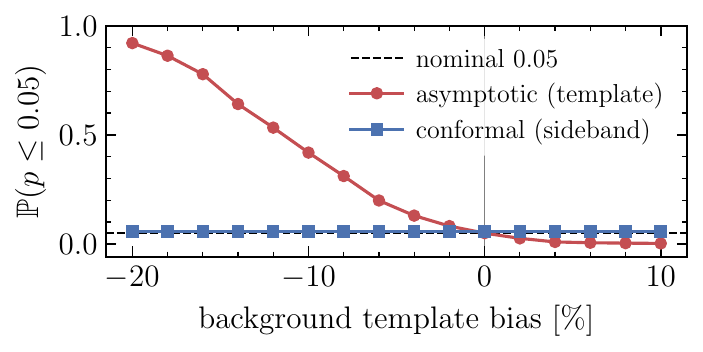}
  \caption{False-positive rate at $\alpha=0.05$ on pure background vs template bias: the asymptotic rate (red) diverges as the template is underestimated, and the conformal rate (blue) stays nominal.}
  \label{fig:lhco_asym}
\end{figure}

\subsection{The role of calibration size}
\label{sec:lhco_calsize}

Figure~\ref{fig:lhco-calsize} studies the dependence on $n_{\rm cal}$ explicitly for both the unweighted and the label-free weighted estimator. At each calibration size we draw $R$ independent background calibration subsets of that size from the sideband reservoir, recompute the background exceedance $\Prob(\hat p\le0.05)$ on the signal-region background for each, and report the mean over draws with an error bar that combines, in quadrature, the across-draw standard deviation (the calibration component) with the binomial test component Eq.~\eqref{eq:se_bin}. The across-draw spread replaces the single-draw bootstrap, which underestimates the scatter at small $n_{\rm cal}$, while the binomial term is the constant floor that survives as $n_{\rm cal}\to\infty$. The weight is held fixed, and only the conformal calibration set is resized. The top panel shows the result. The unweighted exceedance (red) is flat at $\approx0.087$ from $n_{\rm cal}=10^{2}$ to the full $3.8\times10^{5}$. The validity breach does not diminish as calibration data are added, confirming that it is a \emph{systematic} exchangeability bias rather than a finite-sample fluctuation, since the $n_{\rm cal}\to\infty$ limit yields the sideband-tail value rather than the nominal $0.05$. The weighted exceedance (blue) instead tracks the nominal $0.05$ at every size, so weighting restores validity independently of calibration size. At the largest sizes, it plateaus a few $\times10^{-3}$ \emph{above} nominal, and the shrinking error bar leaves $0.05$ about one standard deviation below the point. This is not a failure but the expected imprint of an \emph{estimated} weight. The label-free ratio is learned rather than exact, so a small residual coverage gap survives, bounded by the total-variation distance of Ref.~\cite{Barber2023shift}. It is a fixed estimation residual rather than a divergence; it does not grow with $n_{\rm cal}$, and it only becomes statistically resolvable once the test error has shrunk below it. A more accurate weight model would push it back towards nominal, and for the discovery decision, it is immaterial, being two orders of magnitude smaller than the unweighted breach it replaces.

\begin{figure}[th]
  \centering
  \includegraphics[width=0.95\linewidth]{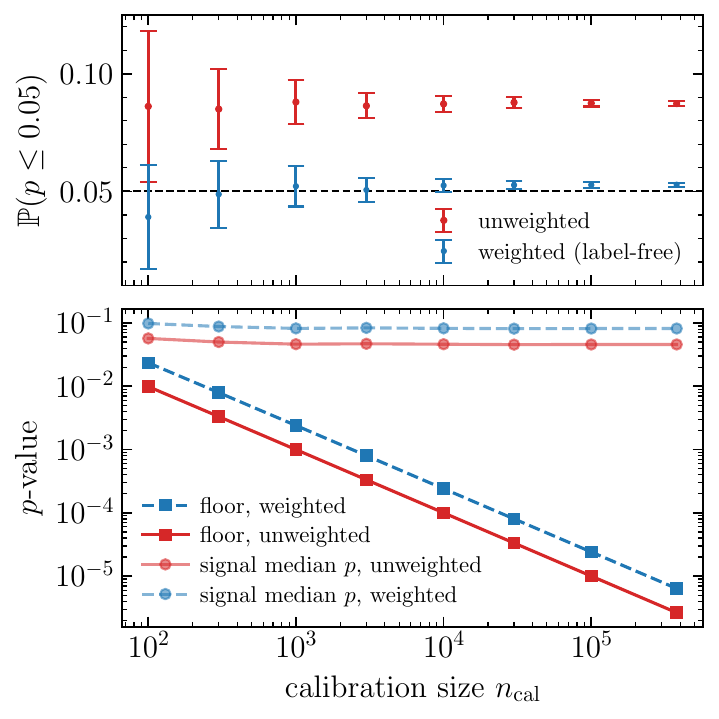}
  \caption{Effect of the calibration-set size $n_{\rm cal}$ for the unweighted (red) and label-free weighted (blue) estimators; points are means over $R$ fresh calibration draws, and the error bars are the across-draw standard deviation combined in quadrature with the binomial test component Eq.~\eqref{eq:se_bin}. \emph{Top:} background exceedance $\Prob(\hat p\le0.05)$ on signal-region background. The unweighted rate is flat at $\approx0.087$ across three orders of magnitude in $n_{\rm cal}$, a systematic exchangeability bias that more data does not remove, whereas the weighted rate sits at the nominal $0.05$ (dashed) at every size, mildly conservative at small $n_{\rm cal}$. \emph{Bottom:} the resolution floor, the smallest resolvable $p$-value, $1/(n_{\rm cal}+1)$ unweighted and $w_{\max}/(w_{\max}+\sum_i w_i)$ weighted (larger by $\approx2.4$), together with the signal median $p$-value (also $\approx2\times$ larger under weighting) that the floor must beat. The weighting that fixes validity costs about a factor of two in both resolving power and signal $p$-value.}
  \label{fig:lhco-calsize}
\end{figure}

The bottom panel shows the resolution floor, the smallest $p$-value a calibration set can assign to the most anomalous event. For the unweighted estimator this is $1/(n_{\rm cal}+1)$; for the weighted estimator a test point beyond all calibration points receives $w(s^\star)/[w(s^\star)+\sum_i w_i]$, so the floor is $w_{\max}/(w_{\max}+\sum_i w_i)$, larger than the unweighted floor by a roughly constant factor $w_{\max}/\langle w\rangle\approx2.4$ set by the up-weighting of the high-score tail. Weighting therefore costs about a factor of two in resolving power: one needs roughly twice the calibration data to reach a given smallest $p$-value. The same panel overlays the signal median $p$-value, which the floor must drop below for the calibration set to resolve the signal; it is likewise about twice as large under weighting, the power cost already seen in Sec.~\ref{sec:lhco_power}. The practical rule follows: size the calibration set by the smallest $p$-value one must resolve, $n_{\rm cal}\gtrsim10^{3}$ for a local $p\sim10^{-3}$ and $n_{\rm cal}\gtrsim3\times10^{6}$ for a discovery-level $p\sim3\times10^{-7}$ ($5\sigma$), doubling that budget if the weighted estimator is used. That budget is a \emph{per-event} statement, the calibration a single event needs to reach a discovery-level local $p$, and it lies an order of magnitude above the $3.8\times10^{5}$ sideband available here; it does not bind the present demonstration, whose quoted significances are counting excesses at the loose threshold $\alpha=0.05$ (Sec.~\ref{sec:lhco_significance}) that accumulate over many events rather than resting on any single one, so the binding requirement is only to resolve $p\sim\alpha$, met with orders of magnitude to spare. With $n_{\rm cal}=3.8\times10^{5}$ the per-event floor is $p_{\min}\approx6\times10^{-6}$ (weighted, $\approx4.4\sigma$), short of a single-event $5\sigma$ but far below the signal's actual median $p$, so the set resolves the signal it is asked to (Fig.~\ref{fig:lhco-calsize}); a per-event $5\sigma$ would require the larger budget above and is a target for a future high-statistics analysis, not a precondition for this study.

\subsection{Null variance of the count field}
\label{sec:nullvar}

The significances in this worked example, the loose-cut count just reported and the scan that follows, are all read off the conformal count statistic $Z(m)$ of Eq.~\eqref{eq:Zfield}, so their honesty rests entirely on using the \emph{correct} standard deviation in the denominator. The purpose of this section is to make sure we do. It is the variance-side counterpart of the calibration check of Sec.~\ref{sec:lhco_validity}. There we corrected the \emph{mean} flag rate, whereas here we correct the \emph{spread} of the flag count, which the binomial formula underestimates because the within-window $p$-values are correlated through their shared calibration set. Too small a denominator would inflate every reported significance, local and global, even with no signal present, so we first measure the size of the effect and then replace the analytic denominator with a resampled one. 

The binomial denominator $n_m\alpha(1-\alpha)$ is exact only if the within-window indicators $\mathbf 1\{\hat p_i\le\alpha\}$ are independent, which they are not. All $p$-values in a window are ranked against a single calibration set, so they are positively dependent, and the true variance is inflated. We measure this inflation directly. Drawing a test window and a calibration set from a common background reservoir, we form the count $k=|\{\hat p\le\alpha\}|$ and estimate the over-dispersion $c=\mathrm{Var}(k)/[n\alpha(1-\alpha)]$ from the scatter of $k$ over many independent draws. Figure~\ref{fig:nullvar-overdispersion} answers how large this inflation is and when it matters. Plotting $c$ against $n_{\rm test}/n_{\rm cal}$, it tracks $c\approx 1+n_{\rm test}/n_{\rm cal}$ (dashed guide), so the binomial formula is exact only when the calibration set dwarfs the window and is anti-conservative by $\sqrt c$ otherwise. For the loose-cut count of Sec.~\ref{sec:lhco_significance}, 
the reservoir dwarfs the test ($c\approx1.16$) and the correction is negligible; 
in the scan below, where each window is calibrated against its neighbouring windows, $c$ ranges from $\approx1.4$ in the three interior windows to $\approx2.7$ at the lower edge, whose single sideband is smaller than the window itself; the naïve significance is then inflated by $\sqrt c$, from $\approx1.2$ in the interior to $\approx1.7$ at that edge.

\begin{figure}[th]
  \centering
  \includegraphics[width=0.95\linewidth]{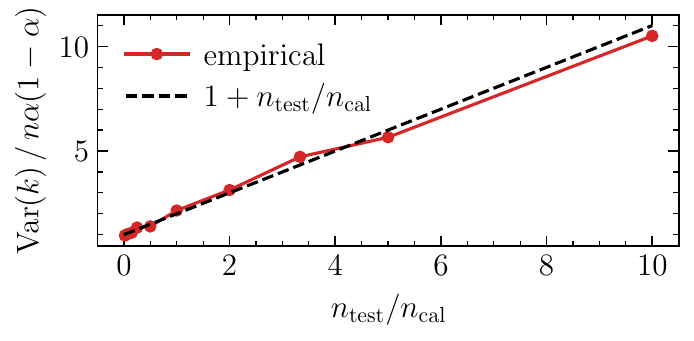}
  \caption{The null variance of the conformal count field: the over-dispersion $c=\mathrm{Var}(k)/[n\alpha(1-\alpha)]$ of the window count, estimated from the scatter of $k$ over many background draws, against $n_{\rm test}/n_{\rm cal}$. It tracks $c\approx1+n_{\rm test}/n_{\rm cal}$ (dashed guide), so the binomial variance is exact only when the calibration set dwarfs the window and understates the true variance by $\sqrt c$ otherwise. The resulting correction removes a systematic, not a statistical, optimism from the reported significance.}
  \label{fig:nullvar}
  \label{fig:nullvar-overdispersion}
\end{figure}

We therefore standardise the scan field by a fully nonparametric null rather than the binomial formula. For each window, the mean $\hat\mu_m$ and standard deviation $\hat\sigma_m$ of $k_m$ are estimated from background-only events, resampling both the calibration band and the test window, and the field is formed as Eq.~\eqref{eq:Zfield}. The effect of this correction on the significance actually quoted on data is shown in the wide-mass scan of Sec.~\ref{sec:lhco_scan} (Fig.~\ref{fig:local-significance}), where the binomial and toy-based standardisations of the conformal count are compared directly. 
Because the toy-based null estimates the mean $\hat\mu_m$ from background only, it re-centres the count as well as rescaling its spread, so it need not order the same way as a rescaled binomial field across windows (made explicit in Sec.~\ref{sec:lhco_scan}). It is this fully nonparametric field, free of the binomial and Gaussian-$c$ approximations, that we feed to the Gross--Vitells correction (Eq.~\eqref{eq:gv}). The correction removes a systematic, not a statistical, optimism from the reported significance.

The Gross--Vitells global $p$-value also depends on a chosen reference level $u$ and the toy-estimated up-crossing rate $\E[N_u]$. We have checked that the global significance of the worked example is stable across $u\in[0,1.5]$ and consistent with the null throughout (Sec.~\ref{sec:lhco_power}), so the result does not hinge on the free choice of $u$; Eq.~\eqref{eq:gv} is an asymptotic bound whose Gaussian-field assumptions are weakest in the extreme tail, so we quote it only as an upper bound and corroborate it with the empirical global tail from background toys.

\subsection{Discovery significance and comparison with the standard workflow}
\label{sec:lhco_scan}
\label{sec:lhco_significance}

\paragraph{A single loose cut.}
Let us first assume that we know where the signal region is and compute the significance without performing a bump hunt. The simplest discovery statistic is a single loose cut on the signal-region score. The setup is a single counting bin, the $[3.3,3.7]$~TeV signal window with no further $m_{jj}$ sub-binning. We compute the conformal $p$-values of the $n_{\rm SR}=6.2\times10^{4}$ signal-region evaluation events at $\alpha=0.05$, calibrated against the full sideband reservoir ($n_{\rm cal}=3.8\times10^{5}$), and the statistic is the count of SR events with $p\le\alpha$ standardised by Eq.~\eqref{eq:Zfield}. The observed count exceeds the nominal expectation $\alpha\,n_{\rm SR}\approx3.1\times10^{3}$, a $Z_{\rm naive}\approx46$ ``excess''. This is almost entirely spurious. It is the $0.087$-versus-$0.05$ calibration breach of Sec.~\ref{sec:lhco_validity} re-expressed as a counting significance and accumulated over $\sim$$6\times10^{4}$ overwhelmingly background events; the injected signal, only $\sim$$0.6\%$ of the SR, cannot account for it, 
and the over-dispersion of Sec.~\ref{sec:nullvar} ($c\approx1.16$ here) lowers it only to $\approx43\sigma$ because the breach is a \emph{bias} in the mean rate, not a variance effect. Applying the label-free weighted calibration returns the background rate to nominal ($0.0499$), and the observed SR count ($k=3116$) matches the corrected expectation ($3103$) to
\begin{equation}
  Z_{\rm corrected}\;\approx\;0.2,
\end{equation}
no significant excess. With a weakly-supervised tree ($\mathrm{SIC}\approx2$) at the native $\sim$$0.6\%$ injection, a single loose cut has no counting reach, and the weighted scan over cuts from $\alpha=0.05$ down to $10^{-3}$ likewise yields $|Z|\lesssim1$. The contrast between $Z_{\rm naive}\approx46$ and $Z_{\rm corrected}\approx0.2$ is the calibration layer in one line. It converts a discovery-faking systematic into an honest null without touching the detector's raw sensitivity. The two numbers are a counting test at a single working point, and the like-for-like comparison with a scanning search is the wide-mass bump hunt below, whose significances are computed differently and are not to be read against the $46$ and $0.2$ here. The size of that gap is itself instructive. This loose cut calibrates the fixed signal-region score against the \emph{full} sideband reservoir, which spans the whole spectrum and is the most mass-mismatched control available, so it carries the largest exchangeability breach, the $0.087$ against the nominal $0.05$. The wide-mass scan instead retrains in each window and calibrates against the two adjacent windows, which are close in mass and only mildly drifted, with a background rate near $0.057$. A counting significance scales as the rate excess times $\sqrt{n}$, so the roughly fivefold larger breach of the full-sideband calibration outweighs the smaller event count and inflates the significance by about a factor of four, which is the difference between the $\sim$$46$ here and the $\sim$$10$--$12$ of the scan. Neither is a signal significance; the true content at this cut is $S/\sqrt{B}\approx3.2$ (Table~\ref{tab:sb}).

\begin{figure*}[t]
    \centering
    \includegraphics[width=0.95\linewidth]{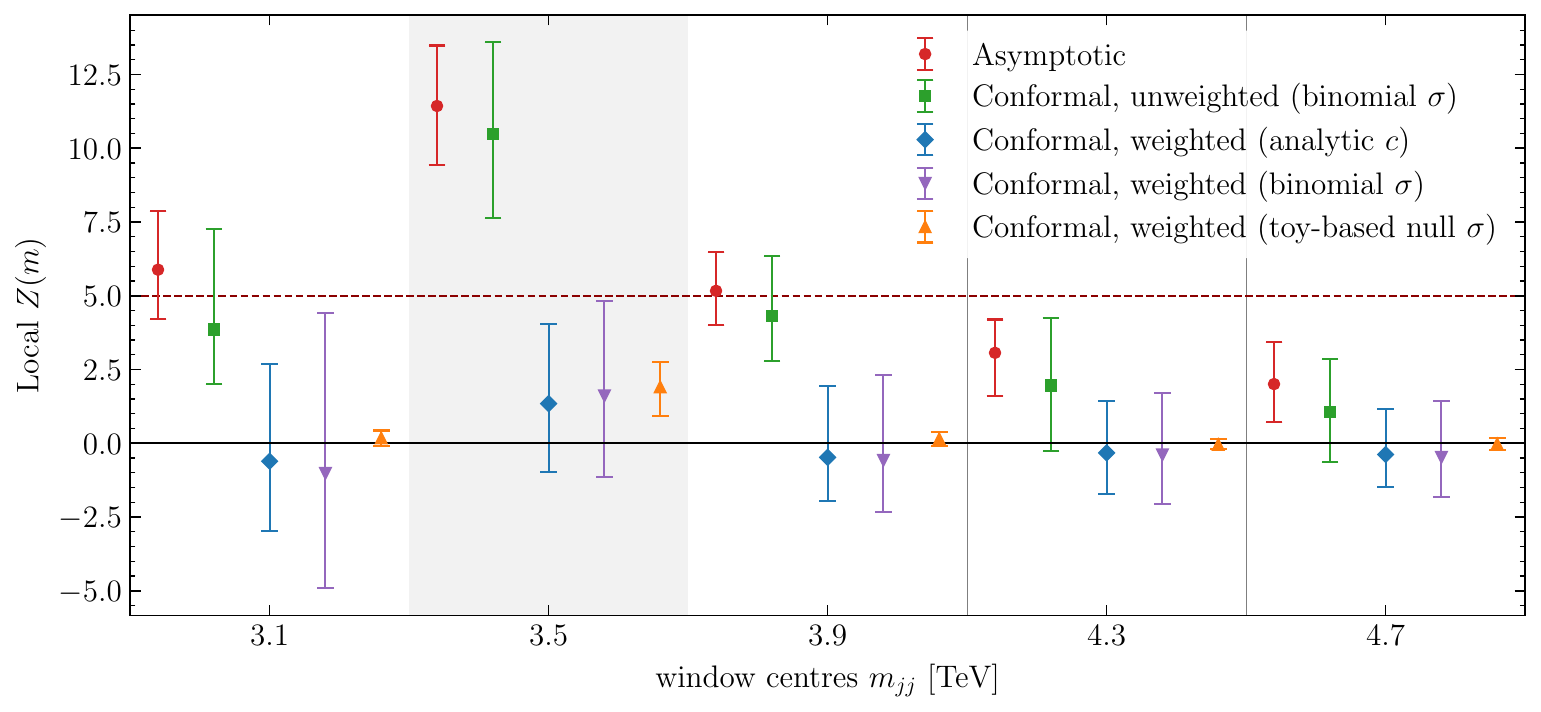}
    \caption{Local significance per window in the blind wide-mass scan, retraining the detector in each of the five windows over $m_{jj}\in[2.9,4.9]$~TeV (the $[3.3,3.7]$~TeV signal window is shaded; the injected resonance sits at $\sim$$3.5$~TeV). Markers are the mean over $50$ random initialisations and the bars their $[16\%,\ 84\%]$-percentile. The asymptotic Asimov significance (red) and the unweighted conformal count (green, binomial standardisation) rise well above $5\sigma$ in several windows, including signal-free ones, an artefact of the score--$m_{jj}$ sculpting amplified by the window statistics. The label-free weighted conformal count is consistent with the null across all windows under all three standardisations: analytic over-dispersion $c$ (blue), binomial (purple), and toy-based null (orange), which is sampled from an additional 100 toys.}
    \label{fig:local-significance}
\end{figure*}

\paragraph{The enhanced bump hunt.}
A single loose cut fixes the signal window in advance. A blind search does not know where the resonance sits, so it scans the spectrum, and a fixed signal-region-trained score is only meaningful at the window it was trained on. We therefore run the scan as a genuine wide-mass search, retraining the detector in each window. The spectrum $m_{jj}~\in[2.9,4.9]$~TeV is split into five windows; in each window, 
that window plays the role of the signal region and its neighbouring windows the sidebands (the three interior windows use both neighbours; the two outermost windows have a single neighbour and are calibrated against that one window alone), and a classifier is retrained out-of-fold (two-fold, so no event is scored by a model trained on it). To average over training stochasticity, we repeat the whole procedure over $50$ random initialisations and report the mean and $[16\%, 84\%]$-percentile. In each window we compute three significance estimates on the identical events: the asymptotic Asimov significance Eq.~\eqref{eq:asimov} at a top-$5\%$ score cut, with the background taken from the sideband efficiency; the unweighted conformal count statistic Eq.~\eqref{eq:Zfield} at $\alpha=0.05$, standardised by the binomial variance; and the label-free weighted conformal count, standardised three ways for comparison, by the analytic over-dispersion $c=1+n_{\rm SR}/n_{\rm SB}$, by the plain binomial variance, and by a fully nonparametric toy-based null built from $100$ background resamples of the window (Sec.~\ref{sec:nullvar}). Because the analytic over-dispersion is $c=1+n_{\rm SR}/n_{\rm SB}$, the single-sided sidebands of the two edge windows make their $c$ larger than the interior ones: with the steeply falling spectrum the populations give $c\approx1.4$ for the three interior windows but $c\approx2.7$ and $c\approx1.5$ at the lower and upper edges, the lower edge being the most populated window calibrated against one smaller neighbour. This enters only the analytic-$c$ and binomial standardisations; the toy-based null we carry forward resamples each window's own sidebands and so absorbs the edge effect automatically, and the local maximum sits in the interior signal window ($c\approx1.4$), so the reported significances are unaffected. The injected resonance sits at $\sim$$3.5$~TeV, i.e.\ only in the second window; the other four are signal-free-ish, so a correct procedure must flag at most the second.


The result is Fig.~\ref{fig:local-significance}. Two features stand out. First, the asymptotic and the unweighted conformal significances do not stay finite and localised. 
They exceed $5\sigma$ not only in the signal window ($\approx11.4$ and $\approx10.5$ respectively) but also in signal-free windows ($\approx5\sigma$ at $3.1$ and $3.9$~TeV), and they grow with the window occupancy. This divergence is unphysical, and it has a precise cause. The score--$m_{jj}$ sculpting of Sec.~\ref{sec:lhco_validity} is a \emph{bias} in the mean rate, which both estimators convert into a significance that scales as $\sqrt{n}$. The asymptotic test trusts a sideband-efficiency background that inherits the sculpting, and the unweighted conformal count uses the binomial variance, which understates the true spread by $\sqrt c$. Neither is a trustworthy discovery significance; both require the correction the conformal layer supplies. Second, the label-free weighted count is consistent with the null in every window under all three standardisations, with only a mild $\approx2\sigma$ at the true signal, the same power cost seen in Sec.~\ref{sec:lhco_power}. Comparing the standardisations of the weighted count is the wide-mass counterpart of the over-dispersion measurement of Sec.~\ref{sec:nullvar}: 
the binomial standardisation (purple) overstates the significance by $\sqrt c$ relative to the analytic-$c$ correction (blue). The two share the nominal centre $\alpha n_m$ and differ only in the variance, inflated by $c$ in the analytic-$c$ case, so the analytic-$c$ correction is the most conservative of the three. The toy-based null (orange) is not a rescaled binomial field; it also re-centres on the background-only mean $\hat\mu_m$, which here sits below $\alpha n_m$, so it returns the \emph{largest} of the weighted significances ($\approx1.9$ local, against $\approx1.6$ binomial and $\approx1.3$ analytic-$c$; Table~\ref{tab:lee}). We carry the toy-based null forward to the trials-factor step because it alone makes no binomial or Gaussian-$c$ assumption, and we report its value precisely because it is the largest, so as not to undersell a possible excess. The qualitative message is the ordering and its cause, not the absolute values. The standard asymptotic procedure fabricates excesses across a blind scan, whereas the weighted conformal layer raises no false alarm anywhere.

The ground-truth signal content puts these numbers in perspective. Table~\ref{tab:sb} lists, per window, the injected-signal count and the naive significance $S/\sqrt{B}$ that a perfect background estimate would give at the same top-$5\%$ cut, with $S$ and $B$ the signal and background yields surviving the cut. The signal reaches $S/\sqrt{B}\approx2.9$ only in the $3.5$~TeV window and is negligible ($\lesssim0.1$) in the other four. 
The asymptotic $\approx11.4$ at the signal window and the $\approx5\sigma$ it reports off resonance therefore overstate the true content by large factors, the direct effect of the sculpted background estimate; the weighted toy-based null ($\approx1.9$) instead sits just below the true $S/\sqrt{B}$, on the conservative side, the price of the label-free weighting. We also provide the SIC values from different samples to demonstrate the robustness of the procedure. Notice that the changes in the value are not due to parameter initialisation but to resampling the training/calibration dataset.

\begin{table}[t]
  \centering
  \setlength{\tabcolsep}{6pt}
  \begin{tabular}{l|ccc}
    \toprule
    window [TeV] & $N_{\rm sig}$ & $S/\sqrt{B}$ & SIC \\
    \midrule
    $(2.9,3.3]$ &     102 & $0.022\pm0.017$ & $0.099\pm0.079$ \\
    $(3.3,3.7]$ &     772 & $2.899\pm1.078$ & $1.308\pm0.486$ \\
    $(3.7,4.1]$ &      82 & $0.136\pm0.061$ & $0.429\pm0.193$ \\
    $(4.1,4.5]$ &       9 & $0.030\pm0.020$ & $0.640\pm0.428$ \\
    $(4.5,4.9]$ &       4 & $0.003\pm0.009$ & $0.108\pm0.325$ \\
    \bottomrule
  \end{tabular}
  \caption{Ground-truth signal content per window at the top-$5\%$ score cut: the injected-signal count $N_{\rm sig}$ in the window and the naive significance $S/\sqrt{B}$ a perfect background estimate would give, with $S$ and $B$ the signal and background yields surviving the cut (truth labels used only for this reference). The signal is concentrated in the $[3.3,3.7]$~TeV window and negligible elsewhere; values are means over $50$ initialisations.}
  \label{tab:sb}
\end{table}

\paragraph{Look-elsewhere correction.}
The five windows are retrained independently and are essentially non-overlapping, so they behave as five independent tests, and the local maximum is penalised for the multiplicity by the exact discrete trials factor
\begin{equation}
  p_{\rm glob} \;=\; 1-(1-p_{\rm local})^{n_{\rm win}}, \qquad n_{\rm win}=5,
  \label{eq:trials}
\end{equation}
the Gross--Vitells up-crossing bound of Eq.~\eqref{eq:gv} being the continuous-scan limit appropriate to a fine scan rather than a handful of independent windows. Table~\ref{tab:lee} collects the local maximum and the corresponding global significance for each estimator. The asymptotic and unweighted-conformal globals remain at $\gtrsim10\sigma$. The trials factor over five windows barely dents a local peak that is itself an artefact, so these would manufacture a discovery out of background sculpting and are not trustworthy. The weighted conformal globals are all consistent with the null, $\lesssim1\sigma$; the toy-based-null value is the one we report, and it is the honest answer for this deliberately weak classifier at the native $S/B\approx0.6\%$ injection.

\begin{table}[t]
  \centering
  \small
  \setlength{\tabcolsep}{5pt}
  \begin{tabular}{lcc}
    \toprule
    Significance estimate & local max $Z$ & global $Z$ \\
    \midrule
    Asymptotic                       & $\approx 11.4$ & $\approx 11.3$ \\
    Conformal, unweighted            & $\approx 10.5$ & $\approx 10.3$ \\
    Conf.\ weighted (analytic $c$)   & $\approx 1.3$  & $\approx 0.3$  \\
    Conf.\ weighted (binomial)       & $\approx 1.6$  & $\approx 0.7$  \\
    Conf.\ weighted (toy null)       & $\approx 1.9$  & $\approx 1.1$  \\
    \bottomrule
  \end{tabular}
  \caption{Look-elsewhere correction for the wide-mass scan of Fig.~\ref{fig:local-significance}. The local maximum significance over the five windows (read from the $50$-initialisation scan) and the global significance after the discrete trials factor Eq.~\eqref{eq:trials}; the conformal entries are the unweighted count (binomial standardisation) and the weighted count under the analytic-$c$, binomial, and toy-based-null standardisations. The asymptotic and unweighted-conformal globals stay at $\gtrsim10\sigma$, an unphysical artefact of the uncorrected sculpting and over-dispersion; the weighted-conformal globals are consistent with the null.}
  \label{tab:lee}
\end{table}

\paragraph{Null calibration of the global significance.}
A small value on this one dataset is not the same as a controlled procedure, so we calibrate the global significance directly on background-only pseudoexperiments. From a signal-free event bank, we draw $500$ calibration pseudoexperiments and run the entire per-window pipeline on each, retraining the score, refitting the label-free weight, and forming the conformal count to build the per-window null distribution of the count. 
A further $500$ test pseudoexperiments, each with an independently regenerated pipeline, are then scored against that null, each yielding a per-window local $p$-value (the empirical upper-tail probability of its count) and a global $p$-value (the empirical tail of the across-window minimum local $p$). Each toy is an $80{,}000$-event subsample drawn without replacement from one bank of $\approx10^{6}$ signal-free events, with the score, weight, and conformal calibration regenerated independently per toy, and the global null taken as the leave-one-out tail of the across-window minimum $p$ over the calibration toys. The bank is finite, so different toys, and the calibration and test sets, reuse events ($\approx8\%$ expected overlap between any two $80{,}000$-event toys); because the pipeline is re-randomised per toy the counts stay nearly independent where the events coincide, and the realised global false-positive rate ($\approx0.056$ at $\alpha=0.05$) with an on-diagonal QQ shows the overlap does not spoil the uniformity at this toy count, with disjoint calibration and test partitions of the bank left as the cleaner high-statistics check. Both are bona fide probabilities in $[0,1]$ by construction. 
Figure~\ref{fig:global-null} shows that the local $p$-values are uniform under the background null and the global $p$-value is close to uniform, so the global false-positive rate is empirically controlled near its nominal level (FPR $\approx0.056$ at $\alpha=0.05$). This is the empirical counterpart to the asymptotic global step in Sec.~\ref{sec:gv}. It provides direct evidence that the weighted-conformal wide-mass procedure does not fabricate global significance when every stage of the analysis is re-randomised. However, it is essential to re-emphasise that the global $p$-values are not sold with coverage guarantees due to the asymptotic part of the procedure. This is only true for the local $p$-values, as demonstrated in Fig.~\ref{fig:global-null}, where certain regions of the global $p$-value are under-covered, so the coverage is only \emph{empirically} satisfied.

\begin{figure}[t]
  \centering
  \includegraphics[width=0.95\linewidth]{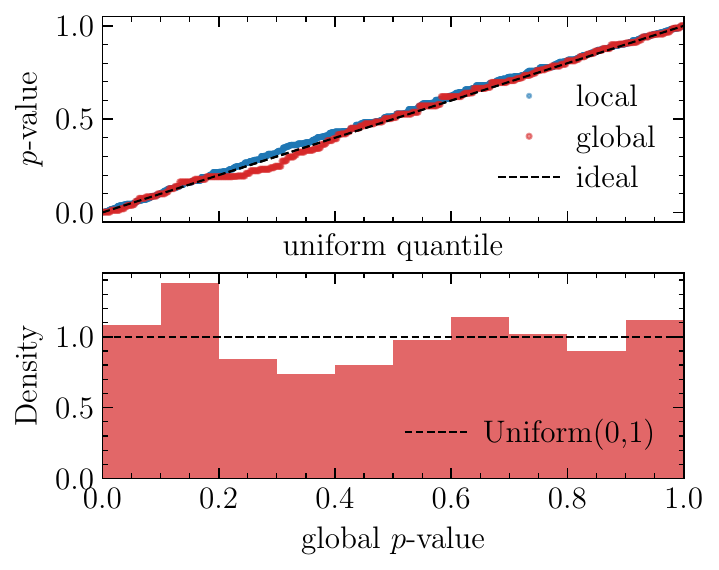}
  \caption{Null calibration of the wide-mass weighted-conformal scan on background-only pseudoexperiments ($500$ calibration and $500$ independent test draws, each regenerating the full per-window pipeline). \emph{Top:} quantile-quantile plot of the per-window local $p$-values (blue) and the global $p$-value (red) against the uniform distribution; the local $p$-values lie on the diagonal and the global $p$-value is close to it. \emph{Bottom:} the global $p$-value histogram, against an ideal uniform distribution (dashed). 
  The local $p$-values are uniform (finite-sample valid); the global $p$-value is close to uniform, so its false-positive rate is empirically controlled near the nominal level, without an analytic coverage guarantee.}
  \label{fig:global-null}
\end{figure}

The comparison with CATHODE is one of calibration, not of raw separation. The fit-independent measure of detector quality is the SIC, on which our deliberately simple tree ($\approx2$) sits well below the CATHODE family ($\approx14$, $11$, $6.5$ for CATHODE, CWoLa hunting and ANODE; Sec.~\ref{sec:lhco_data}), as a weak classifier should. Our point is orthogonal. Whatever the SIC, the asymptotic and unweighted significances inflate with the score--$m_{jj}$ sculpting, whereas the weighted conformal layer reports an honest null and, crucially, a global procedure whose false-positive rate is demonstrably controlled (Fig.~\ref{fig:global-null}).

The one remaining premise is that the local sidebands of each retrained window are background-pure, 
which a residual score--mass correlation mildly violates. We quantify it as the Pearson correlation between the anomaly score and $m_{jj}$ over background events off-resonance ($m_{jj}\notin[3.2,3.8]$~TeV, so the injected signal does not drive it), which gives $r\approx0.17$ (with or without the small signal component); a mass-decorrelated score such as laCATHODE~\cite{Hallin:2022eoq} would reduce both this residual correlation and the over-dispersion and is the natural input to a turnkey wide-mass search. The robust message of Fig.~\ref{fig:local-significance} and Table~\ref{tab:lee} is the contrast between the divergent, uncorrected significances and the calibrated weighted-conformal null, not the absolute values, and it is exactly the progression from the standard asymptotic significance to the auditable, trials-factor-aware significance this paper advocates.

\paragraph{A positive control.} The native injection is too weak for this deliberately poor classifier to discover, so a complete test of the layer must also show that it \emph{can} reach a discovery when a signal is genuinely present. We repeat the wide-mass procedure with a stronger injection, adding events from the pure-signal pool to the signal window at two strengths, $S/B=1.05\%$ and $1.30\%$, up to about twice the native $0.6\%$. Everything else is unchanged from Fig.~\ref{fig:local-significance}, with the same per-window retrained score, label-free weight, toy-based null, and $50$ initialisations. To expose the role of the control, we also leak a fraction of the injected signal into the two adjacent calibration windows, at $5\%$ and $30\%$ of the injected count per window, so the sidebands are no longer perfectly background-pure. Figure~\ref{fig:conf-z-high-sig} shows the per-window significance and Table~\ref{tab:posctrl} the trials-factor-corrected global value.

\begin{figure}[th]
  \centering
  \includegraphics[width=0.95\linewidth]{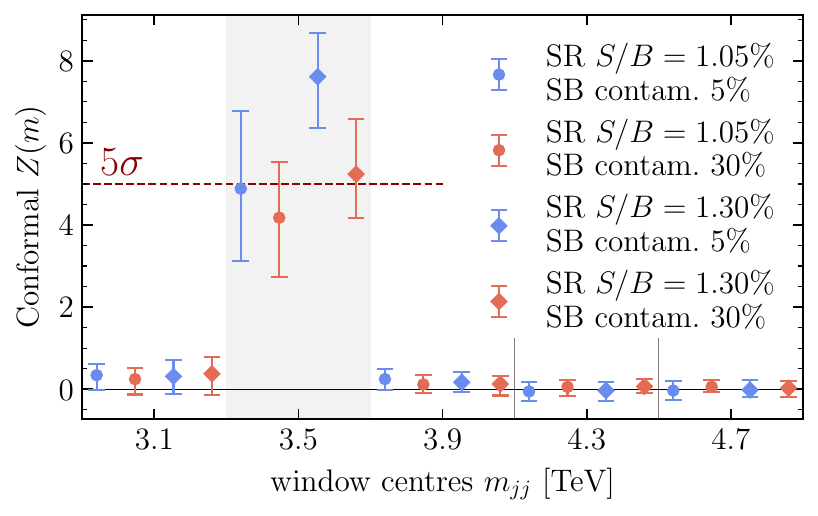}
  \caption{Weighted-conformal local significance $Z(m)$ per mass window for an injected $3.5$~TeV resonance, a positive control for the wide-mass pipeline of Fig.~\ref{fig:local-significance} (per-window retrained score, label-free weight, toy-based null, $50$ initialisations). Four benchmarks cross two signal strengths in the signal region, $S/B=1.05\%$ and $1.30\%$ (the latter about twice the native $0.6\%$ injection), with two levels of signal contamination of the control region, $5\%$ and $30\%$ of the injected signal placed in each adjacent calibration window. Markers are the mean over initialisations and bars the $[16\%,84\%]$ percentiles, with colour encoding the control contamination and marker shape the signal strength. The excess is localised within the shaded signal window, while the far windows remain at the null, so significance is gained where the signal sits rather than across the scan. At low control contamination, the local peak clears $5\sigma$ (dashed); raising the contamination pulls it back towards the null.}
  \label{fig:conf-z-high-sig}
\end{figure}

\begin{table}[t]
  \centering
  \setlength{\tabcolsep}{8pt}
  \begin{tabular}{l|cc}
    \toprule
     & \multicolumn{2}{c}{SB contamination} \\
    \cmidrule(lr){2-3}
    SR $S/B$ & $5\%$ & $30\%$ \\
    \midrule
    $1.05\%$ & $4.6\sigma$ & $3.8\sigma$ \\
    $1.30\%$ & $7.4\sigma$ & $4.9\sigma$ \\
    \bottomrule
  \end{tabular}
  \caption{Trials-factor-corrected global significance of the weighted-conformal chain for the injection study of Fig.~\ref{fig:conf-z-high-sig}, for two signal-region strengths (rows) and two control-region contamination levels (columns). The global $Z$ is the discrete five-window trials factor Eq.~\eqref{eq:trials} applied to the toy-based-null local maximum. 
  }
  \label{tab:posctrl}
\end{table}

The contrast is the message. With a nearly clean control ($5\%$), the weighted chain reaches $7.4\sigma$ global at $S/B=1.30\%$ and $4.6\sigma$ at $1.05\%$, so the layer recovers a true signal rather than only suppressing fakes. Raising the control contamination to $30\%$ pulls the same two points down to $4.9\sigma$ and $3.8\sigma$. Signal that leaks into the calibration window inflates the reference tail and makes the signal-region excess look less anomalous, which is the conservative direction of Theorem~\ref{thm:robust}. Throughout, the excess remains within the signal window, and the far windows sit at the null, so the gain is localised rather than a scan-wide inflation. The clean-control benchmark is not tied to any particular detector. Any neural ansatz that keeps a signal from sculpting its own control region moves the analysis towards it, and the conformal layer reports the significance that follows without further assumptions.

\section{Discussion}
\label{sec:limits}

The proposed methodology does not solve background systematics. Theorem~\ref{thm:wcp} converts a \emph{known or learnable} shift into a correction. An unknown, unparameterised mismodelling remains outside any distribution-free guarantee, conformal or otherwise. The honest scope is that the method makes the assumptions explicit, checkable, and, in the shift and contamination cases, repairable.

\paragraph{Relation to mass-decorrelated detectors.} The exchangeability failure we exploit, jet substructure correlated with $m_{jj}$, is the same pathology that motivated laCATHODE~\cite{Hallin:2022eoq}, which calibrates in a mass-decorrelated latent space so the score does not sculpt the resonant variable by construction. These approaches are complementary, not competing. laCATHODE \emph{prevents} the shift inside the detector, whereas the weighted conformal layer \emph{detects} it through the uniformity diagnostic and \emph{corrects} it post hoc, acting on the score alone and so agnostic to how the score was produced. A decorrelated score would reduce both the over-dispersion and the residual $r\approx0.17$ score--mass correlation documented in Sec.~\ref{sec:lhco_scan}, and is the natural input to a turnkey wide-mass scan. A direct head-to-head comparison on a common benchmark is left to future work. This is orthogonal to a detector's discrimination power. CATHODE's robustness to feature--mass correlations~\cite{Hallin2021cathode} certifies near-optimal \emph{separation}, not the exchangeability of a control region used to calibrate the score, so even a CATHODE score must pass the validity diagnostic before its $p$-values can be trusted.

\paragraph{Signal contamination of the control region.} The injection study of Sec.~\ref{sec:lhco_scan} (Fig.~\ref{fig:conf-z-high-sig} and Table~\ref{tab:posctrl}) isolates one lever on discovery reach, the purity of the control region. With nearly signal-free control, the weighted chain reaches $5\sigma$ for a modest signal, and the reach falls smoothly as the signal leaks into the calibration window. This is not a limitation of the conformal layer, which acts only on the score, but of the detector that decides where the control sits. A score correlated with the resonant variable sculpts its own sidebands, so a localised signal contaminates the very region used to calibrate it. A stronger neural model that suppresses this correlation keeps the control clean and moves the analysis towards the high-reach benchmark of Fig.~\ref{fig:conf-z-high-sig}. The conformal layer is agnostic to how that is achieved. It reports the significance the chosen detector earns and makes the cost of an impure control visible rather than hidden.

\paragraph{Experimental systematics.} The guarantees outlined in this work are statistical. Detector and modelling systematics, jet energy scale and resolution, luminosity, and background shape enter a real analysis through nuisance parameters that the proposed conformal layer does not yet carry. Two partial statements hold. (i) A systematic acting as a smooth covariate shift of the score, for instance, a normalisation or a monotone energy-scale drift between calibration and signal regions, can be corrected by Theorem~\ref{thm:wcp}, given a learnable likelihood ratio. (ii) A bias in the calibration sample itself is the one failure the method cannot self-correct, but, unlike the asymptotic template, it is \emph{visible} as non-uniformity on a trusted control (Sec.~\ref{sec:disc}). A treatment that folds nuisance parameters into the conformal global significance is the most important item for experimental adoption. We will address this in future work.

\paragraph{Scope and limits.} The guarantees presented in this work are precise, yet not bulletproof.
(i) The global $p$-value of Section~\ref{sec:gv} is a Gross--Vitells \emph{bound}, tight at moderate-to-high significance and vacuous at low $z$. It presumes the up-crossing rate is estimated on trustworthy background toys. (ii) Weighted CP (Theorem~\ref{thm:wcp}) needs a
likelihood ratio; a poor estimate degrades coverage by a bounded amount (see ref.~\cite{Barber2023shift}) but not removed. 
(iii) Mondrian calibration costs calibration statistics per bin; very fine binning trades validity for variance. 
(iv) Contamination robustness (Theorem~\ref{thm:robust}) is a non-adversarial result and conservative, costing power. 
(v) The neighbour-calibrated scan field is standardised by a toy-based null (Sec.~\ref{sec:nullvar}). Its validity rests on the local sidebands being background-pure, which a mass-decorrelated score would better ensure (e.g. laCATHODE).
Within these limits, the framework provides an auditable path from an uncalibrated score to finite-sample-calibrated local $p$-values and an asymptotic, trials-factor-aware global significance.

\section{Conclusion \& Outlook}
\label{sec:conclusion}

We have shown that an inductive-conformal layer can sit on top of any anomaly-detection score and turn it into finite-sample-calibrated local $p$-values and an asymptotic, trials-factor-aware global significance, without altering the detector's raw sensitivity. On the LHC Olympics R\&D data, the layer exposed a real exchangeability failure of the sideband calibration, an anti-conservative background rate ($0.087$ against the nominal $0.05$) that, taken at face value in a naïve conformal count, would fake an excess, and the label-free weighted correction restored validity and returned an honest null. Run as a blind wide-mass search that retrains the detector in each window, 
the standard asymptotic and unweighted procedures fabricated $\gtrsim10\sigma$ excesses out of the same sculpting, and $\approx5\sigma$ excesses even in signal-free windows, whereas the weighted-conformal layer returned a global significance consistent with the null, with the false-positive rate of that global procedure verified directly on background-only pseudoexperiments. 
A controlled injection study confirmed the converse. The same weighted chain reaches a $5\sigma$ trials-factor-corrected discovery for a genuine signal at $S/B\approx1.3\%$ once the control region is kept nearly signal-free, so the null on the native sample reflects the deliberately weak classifier and not an insensitive method. The contribution is a calibration and significance layer rather than a new detector, and its value lies in making every assumption explicit and checkable.

Two steps would take the layer from demonstration to deployment. The first is a nuisance-aware calibration that carries detector and modelling systematics through the conformal global significance, the single most important item for experimental adoption. The second is a head-to-head study against mass-decorrelated detectors on a common benchmark, which would quantify how much of the residual sculpting a decorrelated score removes before the conformal layer acts. With those in place, the layer is ready to be folded into a turnkey resonant anomaly search.

\section*{Acknowledgements}
JYA is supported by the Institute for Particle Physics Phenomenology Associateship Scheme and by the Royal Society under grant no. IES/R2/252139.

\bibliography{bibliography}
\end{document}